\documentclass[journal]{IEEEtran}

\ifodd 0
\newcommand{\rev}[1]{{\color{blue}#1}}
\else
\newcommand{\rev}[1]{#1}
\fi

\ifodd 0
\newcommand{\revo}[1]{{\color{blue}#1}}
\else
\newcommand{\revo}[1]{#1}
\fi

\ifodd 0

\else

\fi

\ifodd 0
\newcommand{\revoo}[1]{{\color{purple}#1}}
\else
\newcommand{\revoo}[1]{#1}
\fi

\ifodd 0
\newcommand{\wei}[1]{{\color{purple}#1}}
\else
\newcommand{\wei}[1]{#1}
\fi

\ifodd 0
\newcommand{\weip}[1]{{\color{purple}#1}}
\else
\newcommand{\weip}[1]{#1}
\fi

\ifodd 0
\newcommand{\weipp}[1]{{\color{blue}#1}}
\else
\newcommand{\weipp}[1]{#1}
\fi

\IEEEoverridecommandlockouts

\ifCLASSINFOpdf
\else
\fi
\usepackage{amsmath, amssymb, amsthm, mathrsfs,bm}
\usepackage{graphicx}
\usepackage{amssymb,amsmath,color,graphicx}
\usepackage{subfigure}
\usepackage{cite}
\usepackage{verbatim}
\usepackage{algorithm}
\usepackage{algorithmic}
\usepackage{mathtools}
\usepackage{multirow}

\usepackage{tikz,pgf}
\newtheorem{thm}{Theorem}

\newtheorem{pps}{Proposition}

\newtheorem{asp}{Assumption}
\newtheorem{coro}{Corollary}

\begin{document}

\title{Spectrum Investment under Uncertainty: \\A Behavioral Economics Perspective}

\author{\IEEEauthorblockN{Junlin Yu,~\IEEEmembership{Student Member,~IEEE}, Man Hon Cheung, and Jianwei Huang,~\IEEEmembership{Fellow,~IEEE}}\\
\thanks{Manuscript received May 1, 2016; revised August 5, 2016; accepted August 29, 2016.
This work is supported by the General Research Fund (CUHK 14202814) established under the University Grant Committee of the Hong Kong Special Administrative Region, China. Part of this paper was presented in \cite{yu2014spectrum}.}
\thanks{\revoo{J. Yu, M. H. Cheung, and J. Huang are with the Department of Information Engineering, the Chinese University of Hong Kong, Hong Kong, China;
Emails: \{yj112, mhcheung, jwhuang\}@ie.cuhk.edu.hk.}}
}

\maketitle

\thispagestyle{empty}

\vspace{-2cm}

\begin{abstract}
In this paper, we study a virtual wireless operator's spectrum investment problem under spectrum supply uncertainty. To obtain enough spectrum resources to meet its customer demands, the virtual operator can either sense for the temporarily unused spectrum in a licensed band, or lease spectrum from a spectrum owner. Sensing is usually cheaper than leasing, but the amount of available spectrum obtained by sensing is uncertain due to the primary users' activities in the licensed band.
Previous studies on spectrum investment problems mainly considered the expected profit maximization problem of a risk-neutral operator based on the expected utility theory (EUT).
In reality, however, an operator's decision is influenced by not only the consideration of expected profit maximization, but also the level of its risk preference. 
To capture this tradeoff between these two considerations, we analyze the operator's optimal decision problem using the prospect theory from behavioral economics, which includes EUT as a special case.
The sensing and leasing optimal problem under prospect theory is non-convex and challenging to solve. Nevertheless, by exploiting the unimodal structure of the problem, we are able to compute the unique global optimal solution.
\revoo{We show that comparing to an EUT operator, both the risk-averse and risk-seeking operator achieve a smaller expected profit. On the other hand, a risk-averse operator can guarantee a larger minimum possible profit, while a risk-seeking operator can achieve a larger maximum possible profit. Furthermore, the tradeoff between the expected profit and the minimum possible profit for a risk-averse operator is better when the sensing cost increases, while the tradeoff between the expected profit and the maximum possible profit for a risk-seeking operator is better when the sensing cost decreases.}
\end{abstract}



%

\section{Introduction} \label{sec:intro}
\subsection{Background and Motivation}
\IEEEPARstart{T}{he} business model of \emph{virtual operator}\footnote{Here we focus on a ``virtual operator'' in the wireless industry, which is also referred to as the ``mobile virtual network operator (MVNO)'' in the literature.} has achieved significant success worldwide in recent years \cite{link2}. According to a recent market report published by Transparency Market Research, the global virtual operator market is expected to expand at an annual rate of $7.4$\% and reach a value of US\$ 75.25 billions by 2023 \cite{link3}. 
As a virtual operator (e.g., Consumer Cellular in the US \cite{consumercelluar}) does not own any licensed spectrum, it needs to acquire spectrum from a \emph{spectrum owner} (e.g., AT\&T) in order to provide services to its customers. As virtual operators often rely on leased network infrastructure instead of building and maintaining their own infrastructure, their investment and operational costs are usually lower than the spectrum owners. Such an advantage often enables them to provide cheaper and more flexible data plans to their customers, hence reaching niche markets that are underserved by the spectrum owners \cite{cisco1}.
	
  Motivated by the recent development of cognitive radio technology and dynamic spectrum sharing, a virtual operator can acquire spectrum in two different ways: spectrum sensing and spectrum leasing. 
  With \emph{spectrum sensing} \cite{4796930,4489760}, a virtual operator detects the temporarily unused spectrum in a licensed band, and uses it to provide services to its customers as long as such operator does not cause any harmful interferences to the primary (licensed) customers of the spectrum owner.
  With \emph{spectrum leasing} \cite{6123780,6042865}, a spectrum owner explicitly allows the virtual operator to operate over a given licensed band during a given period of time with a leasing fee.  
  In this paper, we will consider a hybrid spectrum investment scheme involving both approaches. 

  The key feature of the problem is the uncertainty of spectrum acquisition through spectrum sensing, as \weipp{the} virtual operator does not know the activity levels of the primary customers beforehand. When facing uncertainty, most prior studies of spectrum investment applied the \emph{expected utility theory} (EUT) to compute the operator's optimal decisions (e.g., \cite{kasbekar2010spectrum,gao2011spectrum,jin2012spectrum,duan_ia11}). 
  In these models, a (virtual) operator optimizes the decisions to maximize its expected profit.
  Such an EUT model, however, does not fully capture the rather complicated decision process obtained in our daily life, and hence may have a poor predication power \cite{kahneman_pt79}.  
  \rev{Alternatively, the Nobel-prize-winning \emph{prospect theory} (PT) (e.g., \cite{kahneman_pt79, tversky_ai92, kahneman_cv00}), which establishes a more general model than the EUT, provides a psychologically more accurate description of the decision-making under uncertainty. }
PT\footnote{\weipp{To better understand PT, consider the following two lottery settings. Lottery A1: 50\% to win \$200, and 50\% to win \$0; Lottery A2: 100\% to win \$100. Experimental results \cite{kahneman_pt79} showed that most people prefer Lottery A2 to A1. The result reflects that people are risk averse, which we will discuss in more details later in this section. Next we further consider another two lottery settings. Lottery B1: 1\% to win \$99, and 99\% to loss \$1; Lottery B2: 100\% to win \$0. Experimental results \cite{kahneman_pt79} showed that most people prefer Lottery B1 to B2. The result reflects that people will have a subjective probability distortion of small probability events, which we will also introduce later in this section.}} incorporates three main factors in the modeling: 
  (1) \emph{Impact of reference point}: A decision maker evaluates an option based on the potential gains or losses with respect to a reference point, and the choice of the reference point significantly affects the valuation.  
  (2) The \emph{s-shaped asymmetrical value function}: \rev{A decision maker often experiences a diminishing marginal utility when evaluating a gain, and a diminishing marginal disutility when evaluating a loss. Furthermore, it often prefers avoiding losses than achieving gains.} As a result, the value function is s-shaped and asymmetrical: concave in the gain regime, convex in the loss regime, and steeper for losses than for gains (see Fig. \ref{fig:valuepdfun}(a) in Section III-D for a concrete example).
  (3) \emph{Probability distortion}: A decision maker tends to overreact to small probability events, but underreact to medium and large probability events (see Fig. \ref{fig:valuepdfun}(b) in Section III-D for a concrete example). \revoo{This characteristic is useful in explaining behaviors related to lottery and insurance \cite{kahneman_pt79}, where people usually purchase lottery and insurance at prices higher than the expected returns.}
As PT fits better into the reality than EUT based on many empirical studies, researchers and practitioners have applied PT in many areas, such as understanding the behavior of investment agents in finance \cite{jin_bp08} and the effort and wage levels of workers and firms in labor markets\footnote{\weipp{We note that prospect theory is not a theory only about ``irrational behaviors". Instead, it is about how decision makers decide the tradeoff between the maximum/minimum possible profit and the expected profit (in our context), which applies to professionals as well. In fact, there have been various studies (e.g., \cite{dichtl2011portfolio,berejikian2002model}) that focus on the professionals' decisions problems based on prospect theory in areas such as finance and politics.}} \cite{camerer2004behavioral}. However, \weipp{there isn't any} existing work of using PT to understand the spectrum investment behaviors in today's wireless market.

%
%
%
%
%
%
  \subsection{Key Results and Contributions}
  
 In this paper, we study the spectrum sensing and leasing decisions of a virtual operator under sensing \emph{uncertainty}, and formulate it as a two-stage sequential optimization problem.
  In Stage I, the virtual operator determines the optimal amount of licensed spectrum to \emph{sense}.
  Due to the stochastic nature of the primary licensed customers' traffic, the amount of available spectrum obtained through sensing is a random variable. 
  If the spectrum obtained through sensing is not sufficient to satisfy its customers' demand, the virtual operator will \emph{lease} some additional spectrum from the spectrum owner in Stage II.
  
 Under this sensing \emph{uncertainty}, we can model the decision making of a risk-neutral operator, who aims to maximize its expected profit by EUT. 
  However, in reality, a decision maker is rarely risk-neutral. Besides aiming to achieve a high expected profit, it is often affected by its own risk preference. 
  To be more specific, a risk-seeking decision maker is aggressive and wants to achieve a high profit even with a high risk, while a risk-averse decision maker is conservative and wants to guarantee a satisfactory level of minimum possible profit.
  In order to capture this tradeoff between the expected profit and risk preference, we apply the PT to study the optimal sensing and leasing decisions. It leads to a \emph{non-convex} optimization problem, which is very challenging to solve. Nevertheless, by exploiting the unimodal structure of the problem, we can obtain the \emph{global optimal} solution analytically.

	Our key contributions are summarized as follows:

\begin{itemize}

	\item \emph{Behavioral economics modeling of virtual operator's investment decision under uncertainty}: \revoo{We model a virtual operator's investment decisions under sensing uncertainty using PT, which captures the tradeoff between the expected profit maximization and risk preference. We characterize the feature of a risk-averse operator (who is most concerned of potential losses) and a risk-seeking operator (who is most concerned of potential gains).}
	
	\item \emph{Characterization of the unique optimal solution of the non-convex decision problem}: Despite the non-convexity of the spectrum sensing problem, we characterize the uniqueness of the optimal solution and compute it numerically. We further evaluate how different behavioral characteristics (i.e., reference point, probability distortion, and s-shaped valuation) affect this optimal solution.
	
	\item \emph{Engineering insights based on comparison between EUT and PT}:  
	We show that a risk-averse operator can achieve a better tradeoff between the expected profit and minimum possible profit in the high sensing cost scenario than the low sensing cost scenario. The result for the risk-seeking operator is exactly the opposite. 

\end{itemize}

  Next we will review the literature in Section II. In Section \ref{sec:model}, we introduce the spectrum investment model and formulate the sequential optimization problem. In Section \ref{sec:twostage}, we compute the global optimal solution of the non-convex optimization problem, and discuss various engineering insights derived from such a solution. In Section \ref{sec:pdrp}, we illustrate the impact of probability distortion and reference point by considering the special case of binary sensing outcomes. In Section \ref{sec:pe}, we provide simulation results to evaluate the sensitivity of the optimal decision with respect to several model parameters. We conclude the paper in Section \ref{sec:concl}.

\section{Literature Review} \label{sec:review}
\subsection{Expected Profit Maximization in Spectrum Investment Using Expected Utility Theory}
Spectrum investment problem under uncertainty has been studied extensively through expected profit maximization using EUT (e.g., \cite{kasbekar2010spectrum,gao2011spectrum,jin2012spectrum,duan_ia11}).
\revoo{
Kasbekar and Sarkar in \cite{kasbekar2010spectrum} considered a spectrum auction problem under the uncertainty of the number of secondary customers. }
	Gao \emph{et al.} in \cite{gao2011spectrum} studied the spectrum contract between a primary spectrum owner and the secondary customers considering the uncertainty of the customer types.
	In \cite{jin2012spectrum}, Jin \emph{et al.} presented an insurance-based spectrum trading problem between a primary spectrum owner and the secondary customers, where the uncertainty also comes from the types of the customers. 
	Duan \emph{et al.} in \cite{duan_ia11} considered the spectrum investment of a virtual operator under the spectrum sensing uncertainty. 
	In all the above studies based on EUT, a decision maker aims to maximize the \weipp{weighted} average of its utilities under different outcomes, which does not fully capture the realistic human decision behaviors examined in several well known psychological studies in the past few decades \cite{kahneman_pt79, tversky_ai92, kahneman_cv00}. Thus, in this paper we apply the more general PT, which takes into account both the expected payoff and risk preference in the human decision making.

\subsection{Resource Allocation in Communication Networks and Smart Grids Using Prospect Theory}
\revoo{The study of resource allocation in communication networks and smart grids based on PT only emerged recently.} The first paper is due to Li \emph{et al.} in \cite{li_pi12}, which compared the equilibrium strategies of a two-user random access game under EUT and PT.
	Yang \emph{et al.} in \weipp{\cite{yang2014impact}} considered the impact of end-user decision-making on wireless resource pricing, when there is an uncertainty in the quality of service (QoS) guarantees relying on PT.
	Several other recent studies applied PT to study the decision making in smart grid systems. 
In \cite{wang2014integrating}, Wang \emph{et al.} formulated a non-cooperative game among consumers in an energy exchange system. They applied PT to explicitly account for the users' subjective perceptions of their expected utilities. 
    Xiao \emph{et al.} in \cite{xiao2015prospect} studied the static energy exchange game among microgrids that are connected to a backup power plant. They analyzed the Nash equilibria under various scenarios based on PT, and evaluated the impact of user's objective weight on the equilibrium of the game.
	To reflect the fact that realistic decision making is different from expected profit maximization, the studies in \cite{li_pi12, yang2014impact, wang2014integrating, xiao2015prospect} considered a linear value function with the probability distortion. However, to model the impact of risk preference on a decision maker, using only a linear value function is not comprehensive enough. In fact, based on the psychological studies in \cite{kahneman_pt79, tversky_ai92, kahneman_cv00}, the three characteristics of PT (i.e., reference point, s-shaped value function, and probability distortion) together determine the risk preference of a decision maker. Our paper is the first one that considers all three characteristics of PT for a more accurate and comprehensive understanding of the optimal decision problem.

\section{System Model} \label{sec:model}

\subsection{Spectrum Sensing and Leasing Tradeoff}

We consider a cognitive radio network with a spectrum owner and a virtual operator. From the empirical data in \cite{wang2015characterizing,wang2015understanding}, it is possible for the spectrum owner to estimate the spectrum utilization at a particular location accurately based on past measurements. In this paper, we assume that such estimations are accurate, so that the spectrum owner can divide its licensed spectrum into the \emph{primary band} and \emph{secondary band} according to the spectrum utilization. The spectrum owner uses the primary band to serve its primary customers (PCs), while reserves the secondary band to meet the potential leasing requests from the virtual operator. For the virtual operator, it can either try to sense the idle spectrum in the primary band (as PCs' activities are stochastic during the time period of interest), or lease the spectrum from the secondary band.

\begin{figure}[tbp]
\setlength{\abovecaptionskip}{-1mm}
  \setlength{\belowcaptionskip}{-4mm}
\centering\includegraphics[width=0.48\textwidth]{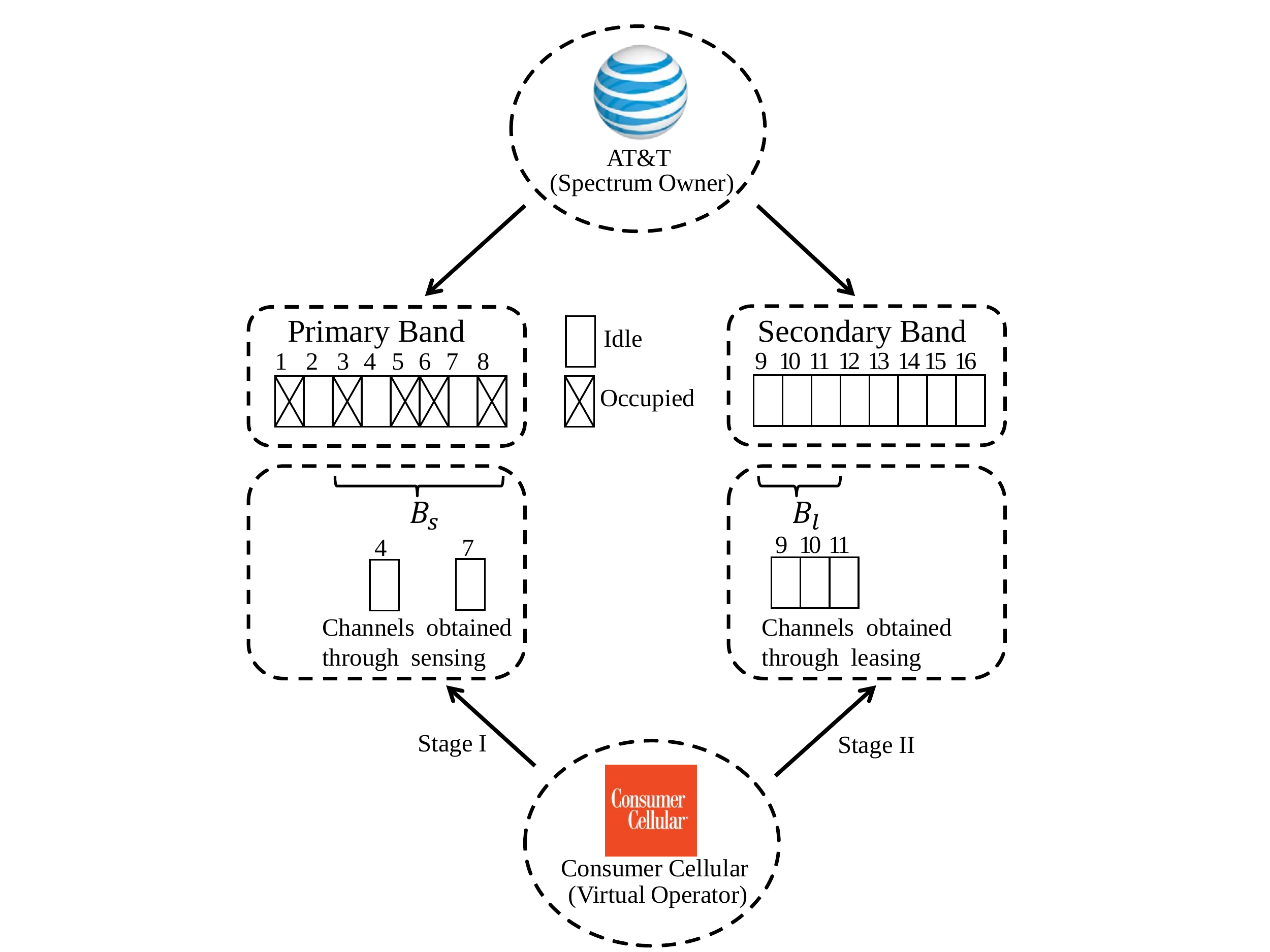}
\caption{An Example of Consumer Cellular and AT\&T.}\label{Fig:1}
\vspace{-4mm}
\end{figure}
	  
As a more concrete (hypothetical) example, we consider the spectrum trading between Consumer Cellular and AT\&T (shown in Fig. \ref{Fig:1}).
  AT\&T is a spectrum owner, who provides wireless services to its PCs. However, it cannot fully utilize its spectrum in some rural areas, so it will divide its spectrum into the primary band (channel 1-8) and secondary band (channel 9-16) at those under-utilized locations.
  Consumer Cellular wants to provide spectrum services to its own customers, but it does not own any spectrum. As a result, Consumer Cellular senses for spectrum holes (not used by the PCs) in the primary band (channel 3-8) without explicit payment to AT\&T, and can also choose to lease spectrum in the secondary band (channel 9-11) with explicit payment to AT\&T.

	From the virtual operator's point of view, sensing is often a cheaper way to obtain spectrum than leasing, because the energy and time overhead involved in sensing is often much lower than the explicit cost of spectrum leasing \cite{duan_ia11}. 
  However, the available amount of spectrum obtained through sensing is uncertain due to the spectrum owner PCs' stochastic activities over time. We would like to understand the virtual operator's optimal spectrum investment decisions in every time slot that strike the best tradeoff between the cost and the risk\footnote{\weipp{We choose the length of time slot such that the primary customers' activities remain unchanged within a time slot \cite{duan_ia11}.}}. 

\subsection{Two-Stage Decision Model}

  We formulate the virtual operator's spectrum investment problem as a two-stage sequential optimization problem in each time slot.
  
	In Stage I (i.e., the sensing stage), the virtual operator determines its sensing decision $B_s$ (measured in Hz). 
	For simplicity, we assume a linear sensing cost $c_s$ per unit of sensed bandwidth, which represents the time and energy overhead for sensing \cite{liang2008sensing}. 
  Due to the stochastic nature of PCs' traffic, only a fraction $\alpha\in [0,1]$ of the sensed spectrum is temporarily available and can be utilized by the virtual operator's own customers. \weip{Hence, the virtual operator obtains a bandwidth of $B_s\alpha$\footnote{\weipp{For simplicity, we assume perfect sensing in this paper. For imperfect sensing \cite{li2014dynamic}, it involves an additional level of uncertainty, which is challenging to consider due to the framing effect \cite{kahneman_pt79} in behavioral economics.}} at the end of the stage. In other words, a large $\alpha$ corresponds to a \emph{high sensing realization}, and a small $\alpha$ corresponds to a \emph{low sensing realization}.} As an example in Fig. \ref{Fig:1}, Consumer Cellular senses six channels in Stage I (i.e., channel 3-8), and only channels 4 and 7 are available. Hence, $\alpha = 1/3$ in this case. We assume that the virtual operator knows the distribution of $\alpha$ through historical sensing results\footnote{\weipp{
  In practice, it is difficult for the operator to obtain the exact distribution of $\alpha$. However, according to \cite{willkomm2008primary, willkomm2009primary, geirhofer2006dynamic, kim2008efficient}, the operator can estimate the distribution of $\alpha$ through learning based on the updated historical sensing results \cite{6185553, 6195522}.}}. \revo{Notice that in the sensing stage, the virtual operator has uncertainty of sensing realization, and its optimal decision will be influenced  by balancing the expected profit and its risk preference on the sensing uncertainty. We will discuss the PT modeling on this aspect in more details in Section III-D. }
    
  In Stage II (i.e., the leasing stage), the virtual operator determines the leasing decision $B_l$ (measured in Hz) after knowing the available amount \weipp{of} spectrum through sensing, $B_s\alpha$. We consider a linear leasing cost $c_l$, which is determined through negotiation between the virtual operator and spectrum owner and is considered to be a fixed parameter in our model. As an example in Fig. \ref{Fig:1}, after Consumer Cellular acquires two channels from sensing (i.e., \weipp{channels} 4 and 7), it further leases three more channels in Stage II (i.e., channel 9-11). \revo{As there is no uncertainty involved in Stage II, there is no difference between EUT and PT modeling in terms of results\footnote{\weipp{In fact, as long as $\lambda=\beta=\alpha=1$, choosing a non-zero value of $R_p$ will just induce a constant shift of the EUT utilities, without affecting the optimal decision under EUT. The reference point affects the analysis of PT, and our analysis shows the impact of reference point in Section V-B.}}.}

  \subsection{Virtual Operator's Profit}

  \revo{When serving its customers, the virtual operator can obtain a revenue of $\pi$ per unit of sold spectrum. We assume that the price $\pi$ is exogenously given and cannot be changed by the virtual operator, due to the intensive market competition \cite{7317774}. Under a fixed usage-based price $\pi$, we assume that the virtual operator's secondary customers' maximum spectrum (bandwidth) demand is $D$\footnote{We do not assume any specific relationship between price $\pi$ and demand $D$ in this paper. Please refer to \cite{7317774} for some further discussions along this line.}. However, the demand may not be fully satisfied if the virtual operator does not have enough spectrum obtained through sensing and leasing discussed before. Hence,}
  the profit of the virtual operator is 
\begin{equation}
R\left(B_s,B_l,\alpha\right) = \pi\min\{D,B_l+B_s\alpha\}-\left(B_s c_s + B_l c_l\right),\label{equ:revenue}
\end{equation}

\noindent
where the revenue (first term on the right hand side) depends on the minimum of the demand $D$ and the spectrum supply $B_l+B_s\alpha$, and the cost (second term on the right hand side) depends on both the sensing decision $B_s$ and leasing decision $B_l$. \revo{As $\alpha$ is a random variable before making the sensing decision $B_s$, we will incorporate the operator's risk preference towards such uncertainty through the modeling based on prospect theory.}

\subsection{Prospect Theory Modeling in Sensing Decision}
\revo{To model the virtual operator's decision under spectrum sensing uncertainty, we consider the following three key features of PT: \emph{reference point} $R_p$, s-shaped \emph{value function} $v\left(x\right)$, and \emph{probability distortion function} $w\left(p\right)$. 

First, the choice of reference point $R_p$ will significantly affect the evaluation of profit $R\left(B_s,B_l, \alpha\right)$ in PT. More specifically, we define the net gain as
\begin{equation}
    x = R\left(B_s,B_l,\alpha\right)-R_p.
\end{equation}
The reference point is a benchmark to evaluate the payoff, where $x \geq 0$ means a gain, while $x < 0$ means a loss. The virtual operator will have different decision mechanisms (to be explained in the next paragraph) when dealing with a gain or a loss in PT. \rev{A higher reference point means that the operator expects a higher profit (at the benchmark), which implies the operator is more risk-seeking.}

Second, as shown in Fig. \ref{fig:valuepdfun}(a), the value function $v\left(x\right)$ is concave for a positive argument (gain), and is convex for a negative argument (loss). Moreover, the impact of loss is usually larger than the gain of the same absolute value. A common choice of value function \cite{tversky_ai92, kahneman_cv00, he2011portfolio} is 
\begin{equation} \label{equ:valuefcn}
     v(x)=\left\{
    \begin{aligned}
    &x^\beta,  &   &\text{if }x\geq 0,\\
    &-\lambda(-x)^\gamma, &   &\text{if }x<0,\\
    \end{aligned}
    \right.\\
\end{equation}

\begin{figure}[t]
\setlength{\belowcaptionskip}{-5mm}
\begin{minipage}{0.48\linewidth}
  \centerline{\includegraphics[width=4cm]{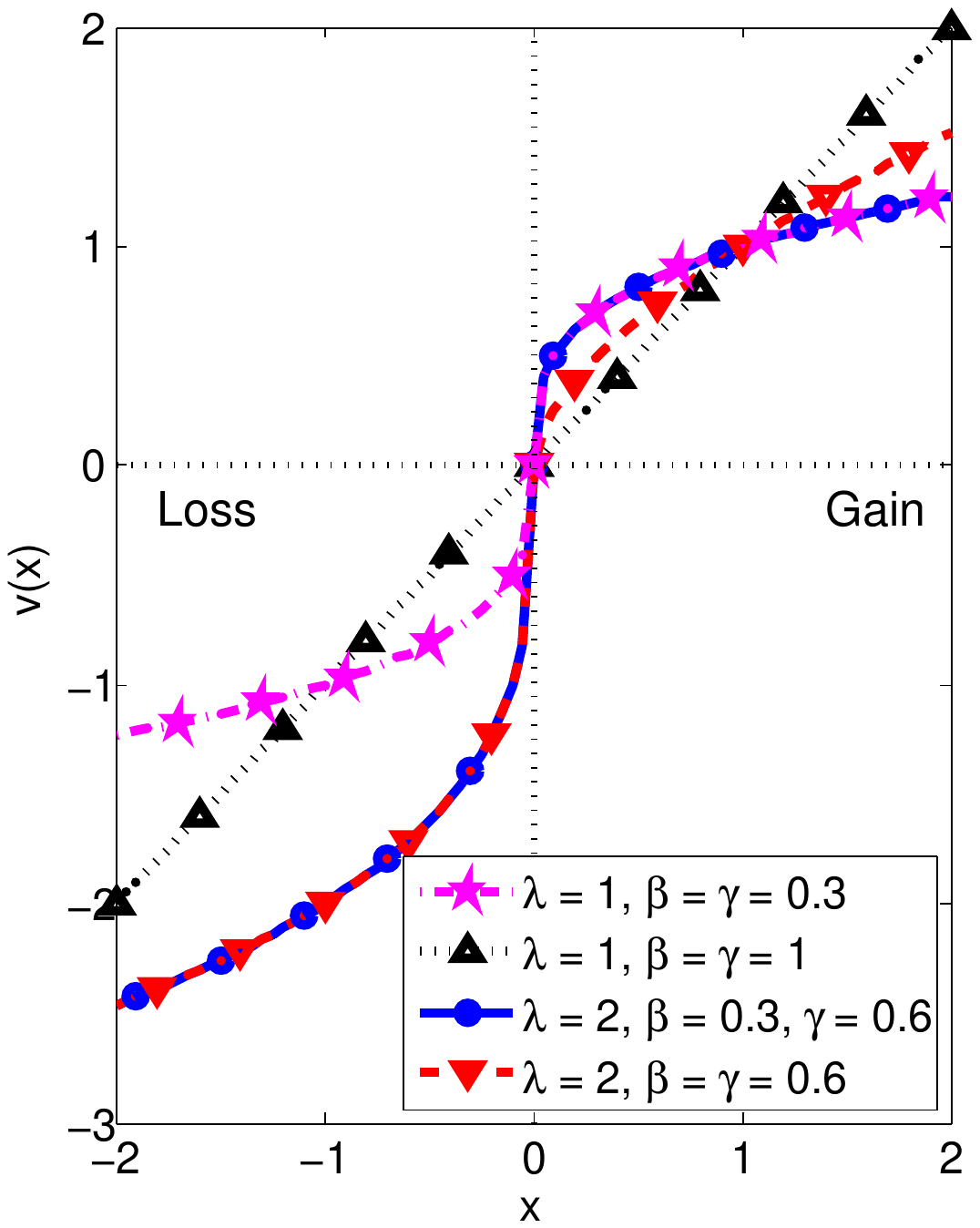}}
  \vspace{-2mm}
  \centerline{\weipp{(a) $v(x)$}}
\end{minipage}
\hfill
\begin{minipage}{.48\linewidth}
  \centerline{\includegraphics[width=4.5cm]{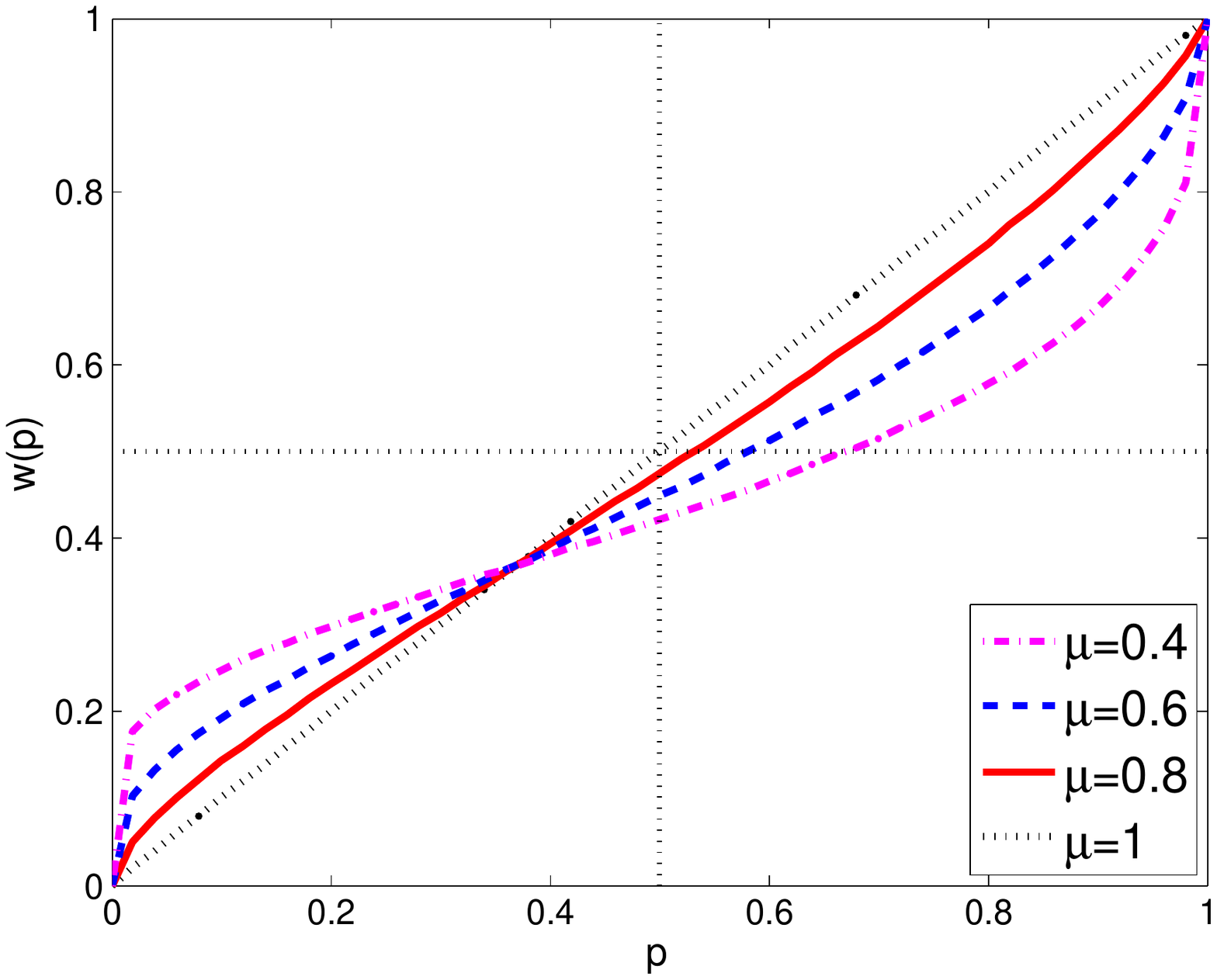}}
      \vspace{-2mm}
  \centerline{(b) $w(p)$}
\end{minipage}
\vfill
\caption{\weipp{The s-shaped asymmetrical value function $v(x)$ and the probability distortion function $w(p)$ in PT.}}
\vspace{-2mm}
\label{fig:valuepdfun}
\end{figure}

\noindent where $\lambda>1$, $0<\beta<1$, and $0<\gamma<1$. 
\rev{The parameter $\lambda$ is the loss penalty parameter, where a larger $\lambda$ indicates that the virtual operator is more concerned of loss, and hence is more \emph{risk-averse}.
The parameters $\beta$ and $\gamma$ are the risk parameters, where the value function of the gain part is more concave (i.e., the virtual operator is more \emph{risk-averse}) when $\beta$ approaches zero, and the value function of the loss part is more convex (i.e., the virtual operator is more \emph{risk-seeking}) when $\gamma$ approaches zero. The impact of $\beta$ and $\gamma$ can be interpreted by the risk-seeking behavior in loss and risk-averse behavior in gain. As an example, a gambler will be more addicted into the gambling when it loses money, and will be less willing to continue when he wins money.}

We note that EUT is a special case if we choose $\lambda = 1$ and $\gamma = \beta=1$. The result under the case $\gamma=\beta$ has been discussed in the conference version of this work in \cite{yu2014spectrum}, and we focus on the more complicated case of $\beta < \gamma$ in this paper\footnote{\revo{The analysis can be readily extended to the case of $\beta>\gamma$, although there are no additional new insights in that case. Hence we omit the discussion of $\beta>\gamma$ here.}}, \weip{which models the scenario that the marginal utility in gain is diminishing faster than the marginal disutility in loss} \cite{rieger2008prospect}.}
\weipp{\begin{asp}
The risk parameter for gain is less than the risk parameter for loss (i.e., $\beta<\gamma$).
\end{asp}}
Third, as shown in Fig. \ref{fig:valuepdfun}(b), the probability distortion function $w(p)$ models the fact that virtual operator overweighs a small probability event and underweighs a large probability event. A common choice of probability distortion function (e.g., \cite{tversky_ai92, kahneman_cv00, he2011portfolio}) is 
\begin{equation} \label{equ:pdfcn}
	w(p) = \exp\left(-\left(-\ln p\right)^\mu\right), \;0<\mu\leq1,
\end{equation}

\noindent where $p$ is the \emph{objective} probability of high sensing realization, and $w(p)$ is the virtual operator's corresponding \emph{subjective} probability. Here the probability distortion parameter $\mu$ reveals how a virtual operator's subjective evaluation distorts the objective probability, where a smaller $\mu$ means a larger distortion. 

\weipp{The key notations of the virtual operator's profit maximization problem are listed in Table \ref{tab:notation}.} In the next section, we will study the virtual operator's optimal sensing and leasing decisions that maximize its profit\footnote{We will simply use ``operator'' to denote ``virtual operator'' for the rest of the paper.}. 

\weipp{
\begin{table}[t]
    \centering
    \vspace{-2mm}
    \caption{Key Notations}
    \vspace{-2mm}
    \begin{tabular}{|c|c|}
    \hline
    \textbf{Symbols} & \textbf{Physical Meanings}\\
    \hline
         $B_s$ (decision variable) & Sensing amount  \\
         \hline
         $B_l$ (decision variable) & Leasing amount \\
         \hline
         $c_s$ & Sensing cost per unit\\
         \hline
         $c_l$ & Leasing cost per unit \\
         \hline
         $\pi$ & Price per unit \\
         \hline
         $D$ & Secondary users' demand \\
         \hline
         $\alpha$& Sensing realization factor \\
         \hline
         $\beta$& Risk parameter for gain \\
         \hline
         $\gamma$& Risk parameter for loss \\
         \hline
         $\lambda$& Loss penalty parameter \\
         \hline
         $\mu$& Probability distortion parameter \\
         \hline
    \end{tabular}
    \label{tab:notation}
\end{table}
}

\section{Solving the Two-stage Optimization Problem} \label{sec:twostage}

  In this section, we use backward induction \cite{mas1995microeconomic} to solve the two-stage sequential optimization problem.  
  In Section IV-A, we derive the operator's optimal leasing decision in Stage II.
  In Section IV-B, we obtain the operator's optimal sensing decision in Stage I under sensing uncertainty by the PT model\footnote{It should be noted that the models under EUT and PT share the same Stage II (i.e., the leasing stage), because uncertainty only appears in Stage I (i.e., the sensing stage).}.

\subsection{Optimal Leasing Decision in Stage II}
In Stage II, given a fixed value of sensing bandwidth $B_s$ determined in Stage I, the operator's leasing optimization problem is \begin{align} \label{equ:problem_lease}
\max_{B_l\geq 0}~~ R(B_s,B_l,\alpha)\!=\!\pi\min\{D,(B_l\!+\!B_s\alpha)\}\!-\!(B_s c_s\! + \!B_l c_l).
\end{align}  

In our analysis, we make the following assumptions.
\vspace{-2mm}
\weipp{\begin{asp}
 (a) The sensing cost is less than the leasing cost (i.e., $c_s < c_l$); (b) The
leasing cost is less than the operator’s usage-based price (i.e., $c_l<\pi$).
\end{asp}}
\vspace{-2mm}
Both parts of Assumption 2 allow us to focus on the non-trivial case in our analysis. When $c_s\geq c_l$, we can show that the operators' optimal decision of the two-stage process is always \weipp{$B_s^* = 0$ and $B_l^* = D$} \revo{for any type of risk preferences, because leasing is both cheaper and risk-free.} 
When $c_l \geq \pi$, we can show that the operator will always choose $B_l^* = 0$ \revo{for any type of risk preferences}. Hence, we can focus on Assumption 2 without loss of any generality. Under Assumption 2, we can show that the optimal leasing decision that \weipp{solves} Problem \eqref{equ:problem_lease} is
\vspace{-2mm}
\begin{equation}
B_l^*=\max \{D-B_s\alpha, 0\}, \label{wei:6}
\vspace{-2mm}
\end{equation} 

\noindent
which is the difference between the total demand $D$ and the available spectrum through sensing, $B_s\alpha$. If $B_s\alpha$ exceeds $D$, then $B_l^* = 0$.
Under \eqref{wei:6}, the operator's profit in (1) can be written as a function of $B_s$ and $\alpha$:
\vspace{-0.8mm}
\begin{equation}
R(B_s,B_l^*,\alpha)=\pi D-B_s c_s - \max \{D-B_s\alpha, 0\} c_l.\label{equ:revenues2}
\end{equation}
\vspace{-5mm}

\subsection{Optimal Sensing Decision in Stage I}
\revo{We assume that the sensing realization factor $\alpha$ follows a discrete distribution with $I$ possible outcomes\footnote{\weipp{Since the spectrum is usually divided into a finite number of channels in practical systems, it is reasonable to consider a discrete distribution of $\alpha$.
Mathematically, when we choose the number of possible realizations $I$ to be large, then the discrete distribution can well approximate a continuous distribution.}}, hence has a finite number of sensing realization outcomes \cite{zeng2010review, li2014dynamic}}. The corresponding probability mass function is denoted as
	\begin{align}
	\vspace{-2mm}
	p(\alpha_i) \triangleq\mathbb{P}(\alpha=\alpha_i)=p_i,\quad i\in \mathcal{I} = \{1,...,I\}.
	\vspace{-2mm}
	\end{align}
Without loss of generality, we assume that $\alpha_i<\alpha_j$ if $i<j$.
By substituting (2), (3), (4) and (8) into (7),
  we obtain the \emph{expected utility} of the operator under PT as 
%
\begin{align}
U&\left(B_s\right)=\sum_{i = 1}^I v\left[R\left(B_s, B_l^*, \alpha_i\right) - R_p\right] w\left(p\left(\alpha_i\right)\right) \notag\\
=& \sum\limits_{i=1}^I v\left[\pi D\!-\!\left(B_s c_s \!+\! \max \{D\!-\!B_s\alpha_i, 0\} c_l\right)\!-\!R_p\right]w\left(p\left(\alpha_i\right)\right).
\end{align}


  The operator's spectrum sensing optimization problem in Stage I with sensing uncertainty is
\begin{align} \label{equ:problem_sense}
\max_{B_s\geq 0}\quad U(B_s).
\end{align}  

\begin{table*}[t] 
\centering 
\vspace{-2mm}
\caption{Optimal Sensing and Leasing Decision under PT}
\vspace{-2mm}
\begin{tabular}{|c|c|c|}
\hline
\textbf{Condition}&\textbf{Optimal Sensing Decision} $B_s^*$&\textbf{Optimal Leasing Decision} $B_l^*$\\ \hline
$D\leq M_{{\hat{\imath}}+1}$&$B_s^*=\frac{D}{\alpha_{{\hat{\imath}}+1}}$&$B_l^*=\max\{0,D-\frac{D\alpha}{\alpha_{{\hat{\imath}}+1}}\}$\\ \hline
$M_j< D< H_{j}$ for $j={\hat{\imath}}+1,...,I-1$&$B_s^*=g_{j}^{-1}(0)$&$B_l^*=\max\{0,D-\alpha g_{j}^{-1}(0)\}$\\ \hline
$H_{j}\leq D< M_{j+1}$ for $j={\hat{\imath}}+1,...,I-1$&$B_s^*=\frac{D}{\alpha_{j+1}}$&$B_l^*=\max\{0,D-\frac{D\alpha}{\alpha_{j+1}}\}$\\ \hline
$D\geq M_I$&$B_s^*=\frac{M_I}{\alpha_I}$&$B_l^*=D-\frac{M_I\alpha}{\alpha_I}$\\ \hline
\end{tabular} \label{table:pt}
\vspace{-4mm}
\end{table*}
  
\subsubsection{\weip{Optimal Sensing Decision under a Risk-free Reference Point $R_p=D(\pi-c_l)$}}

Problem \eqref{equ:problem_sense} is a \emph{non-convex} optimization problem, because it involves the s-shaped value function in \eqref{equ:valuefcn}. Hence, it is challenging to analytically characterize the closed-form optimal solution. However, we can show that there exists a \emph{unique} global optimal solution by exploiting the special \emph{unimodal} structure of the problem. 

In problem \eqref{equ:problem_sense}, a common choice of reference point is the risk-free profit\footnote{\weip{
A complete analysis for an arbitrary reference point is challenging because of the non-convexity in Problem \eqref{equ:problem_sense}. In Section V, we will further look into the impact of reference point in a simplified model with binary sensing results.
}}. For example, in finance, investors naturally
choose the risk-free return as a benchmark to evaluate their investment performances \cite{he2011portfolio}. Here, we choose the maximum profit that the operator can achieve without sensing (hence a risk free choice) as the reference point. This corresponds to the operator leasing a bandwidth $B_l=D$ and choosing $B_s = 0$, which leads to the profit of
\begin{equation}
R_p=D(\pi-c_l).\label{wei:11}
\end{equation}

Substituting \eqref{wei:11} into \eqref{equ:problem_sense}, we solve problem \eqref{equ:problem_sense} and summarize the key result in Table \ref{table:pt}. \revo{To understand Table \ref{table:pt}, we first define some notations.}
\rev{We define $\hat{\imath}$ $\in \mathcal{I}$ as the unique index that satisfies both the constraints $\alpha_{{\hat{\imath}}}\leq \frac{c_s}{c_l}$ and $\alpha_{{\hat{\imath}}+1}> \frac{c_s}{c_l}$\footnote{Both the case $\alpha_1\geq\frac{c_s}{c_l}$ and the case $\alpha_I\leq\frac{c_s}{c_l}$ are trivial, and the discussion under these two cases \weipp{are} covered in the discussions under the case $\alpha_1<\frac{c_s}{c_l}<\alpha_I$.}.
We define $H_j$ (for $j={\hat{\imath}}+1$, ..., $I-1$) (see (17) in Appendix A) and $M_j$ (for $j={\hat{\imath}}+1$, ..., $I$) (see (15) and (16) in Appendix A) in Table \ref{table:pt} as decision indicators. The decision indicators are increasing functions of the leasing cost $c_l$ and the risk parameter for gain $\beta$, and decreasing functions of the sensing cost $c_s$, the risk parameter for loss $\gamma$, and the loss penalty parameter $\lambda$. In other words, when the sensing cost $c_s$ is higher or when the operator is more risk-averse (e.g., with a larger $\lambda$, a smaller $\beta$, or a larger $\gamma$), the decision indicators $H_j$ and $M_j$ decrease. The function $g_j(B_s)$ (for $j=\hat{\imath}+1,...,I-1$) (see (27) in Appendix A) is a decreasing function in $B_s$. }


\weipp{
\begin{thm} \label{thm:pt_table}
Under Assumptions 1 and 2, the optimal sensing decision $B_s^*$ for problem (10) and the optimal leasing decision $B_l^\ast$ for problem (5) are summarized in Table \ref{table:pt}. 
\end{thm}}
	
 \weip{By using the unimodal structure of Problem \eqref{equ:problem_sense}, we show that there is at most one inner local maximum point, and hence the global \weipp{optimum} is either at the local maximum point or at the boundaries \cite{dharmadhikari1988unimodality}. For the detailed proof of Theorem \ref{thm:pt_table}, please refer to Appendix \ref{app:pt_table}. 
 As shown in Table \ref{table:pt}, the operator's optimal sensing and leasing decisions depend on the decision indicators $H_j$ and $M_j$.}
 Notice that $M_j<H_j<M_{j+1}$ for every $j$, and all $H_j$ (for $j={\hat{\imath}}+1$, ..., $I-1$) and $M_j$ (for $j={\hat{\imath}}+1$, ..., $I$) are increasing in $c_l$ and $\beta$, and are decreasing in $c_s$, $\gamma$, and $\lambda$. For a more risk-averse operator (with larger $\lambda$, smaller $\beta$, and larger $\gamma$), it has smaller decision indicators $H_j$ and $M_j$ (which refers to a lower row in Table \ref{table:pt}), hence it leads to a smaller $B_s^*$.

  To better illustrate the insights behind Table \ref{table:pt}, we further characterize the impact of sensing cost $c_s$ and the risk parameters $\lambda$, $\beta$ and $\gamma$ on $B_s^*$ in the following corollaries. 

\weipp{
\begin{coro} \label{lem:lambda}
 The optimal sensing decision $B_s^*$ for problem (10) is decreasing in the loss penalty parameter $\lambda$ in (3) and the risk parameter for loss $\gamma$ in (3).
\end{coro}}
\weip{The proof of Corollary \ref{lem:lambda} is given in Appendix B}. Corollary \ref{lem:lambda} indicates that if an operator is more risk-averse (with a larger loss penalty parameter $\lambda$ or a larger risk parameter for loss $\gamma$), it will sense less, as it prefers avoiding loss due to a low sensing realization to achieving gain due to a high sensing realization. 

\weipp{\begin{coro}\label{lem:beta}
 When the demand is larger than the $I^{th}$ decision indicator (i.e., $D>M_I$), the optimal sensing decision $B_s^*$ for problem (10) is increasing in the risk parameter for gain $\beta$ in (3).
\end{coro}}
\weip{The proof of Corollary \ref{lem:beta} is given in Appendix C}. 
Corollary \ref{lem:beta} indicates that for the case of a large demand (i.e., $D>M_I$), a more risk-seeking operator (with a larger risk parameter for gain $\beta$) will sense more to achieve a larger gain. 

However, in the small demand case (i.e., $D\leq M_I$), as long as the total available spectrum from sensing can satisfy the demand (i.e., $B_s\alpha\geq D$), a larger sensing decision $B_s$ leads to a constant revenue $D\pi$ but a larger total sensing cost $B_sc_s$, hence a less profit. Therefore, the optimal sensing decision $B_s^*$ is not always increasing with $\beta$ when $D \leq M_I$.

\weipp{\begin{coro}\label{lem:cs}
The optimal sensing decision $B_s^*$ for problem (10) is decreasing in the sensing cost $c_s$. 
\end{coro}}
\weip{The proof of Corollary \ref{lem:cs} is given in Appendix D}. Corollary 3 indicates that the operator is more willing to sense when the sensing cost decreases.

For comparison, we also consider the operator's optimal sensing and leasing decisions under the EUT model. \weip{The detailed discussion is given in Appendix E}.

\section{Special Case: Optimal Sensing Decision with Binary Outcomes} \label{sec:pdrp}
In Section IV, we have focused on the discussion of the impact of the parameters of the s-shaped value function. In this section, we consider a special case with binary sensing outcomes (i.e., $I=2$), to further illustrate the impact of reference point and probability distortion on the sensing decision. More specifically, the binary results are $\alpha_1=0$ and $\alpha_2=1$ with probabilities $p_1$ and $p_2$ as in (8) respectively, where $p_1+p_2=1$.    

A practical motivation of this case is the spectrum utilization at those areas such as the subway stations or overpasses, where the traffic patterns of PCs follow a very high peak-valley ratio \cite{wang2015characterizing,wang2015understanding}. In these areas, there is hardly any spectrum left for the virtual operator when the traffic is at its peak, but the situation completely changes during the off-peak hours.
\revoo{Another example is the spectrum used by a rotating radar system \cite{6331681}, where the radar antenna gain and radiation pattern seen by the operator vary in a very high peak-valley ratio with the rotation of main beam. Hence, there are periods of time \weipp{where} the sensing realization is very high}, and other periods of time when the sensing realization is very low.
\weip{To better illustrate the impact of reference point and probability distortion, we consider the case $\beta=\gamma$ in this section to simplify the analysis\footnote{\weip{We do not obtain additional new insights on reference point and probability distortion in the case $\beta\neq\gamma$, hence we omit the discussion of $\beta\neq\gamma$ here.}}.}
Hence, by plugging $I=2$ into (9), we can obtain the utility function:
\begin{align}
U&(B_s)=w(p_1)v(\pi D-B_sc_s-Dc_l-R_p)\notag\\
+&w(p_2)v(\pi D\!-\!B_sc_s\!-\!\max\{0,D\!-\!B_s\}c_l\!-\!R_p).
\end{align}

\subsection{Impact of Probability Distortion}  
First, we discuss the impact of probability distortion on the operator's optimal sensing decision $B_s^*$. To compute the optimal sensing decision, we consider the same reference point $R_p=D(\pi-c_l)$ as in Section IV-B, and obtain the following theorem.

\weipp{
\begin{thm}
For the case of binary outcomes, under Assumptions 1 and 2, the optimal sensing decision $B_s^*$ for problem (10) is summarized in Table \ref{table:pdd}. 
\end{thm}}

\begin{table}[t]  
\centering 
\vspace{-2mm}
\caption{Optimal Sensing Decision under different $\mu$}
\vspace{-2mm}
\begin{tabular}{|c|c|}
\hline
\textbf{Condition}&\textbf{Optimal Sensing Decision} $B_s^*$\\ \hline
$\frac{c_l}{c_s}\leq\left(\frac{w\left(p_1\right)}{w\left(p_2\right)}\right)^{\frac{1}{\beta}}+1$&$B_s^*=0$\\ \hline
$\frac{c_l}{c_s}>\left(\frac{w\left(p_1\right)}{w\left(p_2\right)}\right)^{\frac{1}{\beta}}+1$&$B_s^*=D$\\ \hline
\end{tabular} \label{table:pdd}
\vspace{-4mm}
\end{table}

The proof of Theorem 2 is given in Appendix F. The results in Table \ref{table:pdd} depend on the ratio $w(p_1)/w(p_2)$. The operator will sense $B_s^*=D$ when the leasing and sensing cost ratio ${c_l}/{c_s}$ exceeds the sensing threshold $\left[w\left(p_1\right)/w\left(p_2\right)\right]^{\frac{1}{\beta}}+1$, and will sense $B_s^*=0$ otherwise. Since $p_1+p_2=1$, according to the probability distortion effect in PT, the larger probability is underweighed and the smaller probability is overweighed. From Table \ref{table:pdd}, for the case $p_1<p_2$, $p_1$ is overweighted and $p_2$ is underweighted in PT (i.e., $w(p_1)/w(p_2)>p_1/p_2$), where the operator \emph{underestimates} the chance of having a high sensing realization. 
Thus, the PT operator is more risk-averse and will sense less than or equal to an EUT operator because it has a higher sensing threshold. On the other hand, when $p_1>p_2$, $p_2$ is overweighted (i.e., $w(p_1)/w(p_2)<p_1/p_2$), where the operator \emph{overestimates} the chance of a high sensing realization. Thus, the PT operator is more risk-seeking and will sense larger than or equal to that of an EUT operator because the PT operator has a lower sensing threshold.

\subsection{Impact of Reference Point}  
Next, we discuss the impact of reference point $R_p$ on the operator's optimal sensing decision $B_s^*$. A high reference point $R_p$ indicates that the operator has a high expectation on the profit, and it is more likely to experience a \emph{loss} since the outcome is often less than its expectation. On the other hand, a low reference point $R_p$ indicates that the operator has a low expectation, and it is more likely to experience a \emph{gain} since the outcome is often beyond its expectation. As we will see in the following, whether an outcome is considered as a loss or a gain can significantly affect the operator's subjective valuation of the outcome, and hence its sensing decision.

To better illustrate the impact of reference point, we focus on two choices: (i) A high reference point $R_p^{\text{high}}=D(\pi -c_s)$, which reflects the operator's expectation of realizing all the sensing spectrum $D$. (ii) the low reference point $R_p^{\text{low}}=D(\pi-c_l-c_s)$, which reflects the operator's expectation of not realizing any of the sensing spectrum $D$. In other words, the same outcome is more likely to be considered as a loss under the high reference point than under the low reference point. 

From (12), we obtain the utility under the high reference point as
\begin{align}
U_{R\!P\!H}\left(B_s\right)\!=&\!-\lambda w\left(p_1\right)\left(B_sc_s+Dc_l-Dc_s\right)^\beta\notag\\
&-\lambda w\left(p_2\right)\left[D\left(c_l\!-\!c_s\right)-B_s\left(c_l\!-\!c_s\right)\right]^\beta,
\end{align}
and the utility under the low reference point as
\begin{align}
U_{R\!P\!L}\!\!\left(\!B_s\!\right)\!\!=\!\!w\!\left(\!p_1\!\right)\!\left(\!Dc_s\!\!-\!\!B_sc_s\!\right)^\beta\!\!\!+\!w\!\left(\!p_2\!\right)\left[B_s\left(\!c_l\!-\!c_s\!\right)\!\!+\!\!Dc_s\right]^\beta\!\!.
\end{align}

By studying the first order derivatives under the two reference points $U_{RPH}'(B_s)$ and $U_{RPL}'(B_s)$, we can compute the optimal sensing decision that solves problem (10) in the following theorem.
\weipp{
\begin{thm}
For the case of binary outcomes, under Assumptions 1 and 2, the optimal sensing decision $B_s^*$ for problem (10) under different reference points $R_p^{high}=D(\pi-c_s)$ and $R_p^{low}=D(\pi-c_l-c_s)$ are summarized in Table \ref{table:rpp}. 
\end{thm}
}
\begin{table}[t]  
\centering 
\vspace{-2mm}
\caption{Optimal Sensing Decision $B_s^*$ under $R_p^{\text{high}}=D(\pi\!-\!c_s)$ and $R_p^{\text{low}}=D(\pi\!-\!c_l\!-\!c_s)$}
\vspace{-2mm}
\begin{tabular}{|c|c|c|}
\hline
\multirow{2}{*}{\textbf{Condition}}&$B_s^*$ \textbf{under} $R_p^{\text{high}}$  &$B_s^*$ \textbf{under} $R_p^{\text{low}}$ \\ 
&(i.e., Risk-Seeking)&(i.e., Risk-Averse)\\\hline
$\frac{c_l}{c_s}\!<\!1\!+\!\frac{w(p_1)}{w(p_2)}$&$B_s^*=U_{RPH}^{'-1}(0)$&$B_s^*=0$\\ \hline
$\frac{c_l}{c_s}\!\geq\! 1\!+\!\frac{w(p_1)}{w(p_2)}$&$B_s^*=D$&$B_s^*=U_{RPL}^{'-1}(0)$\\ \hline

\end{tabular} \label{table:rpp}
\vspace{-4mm}
\end{table}

\begin{figure*}[t]
\begin{minipage}{0.32\linewidth}
  \centerline{\includegraphics[width=5.5cm]{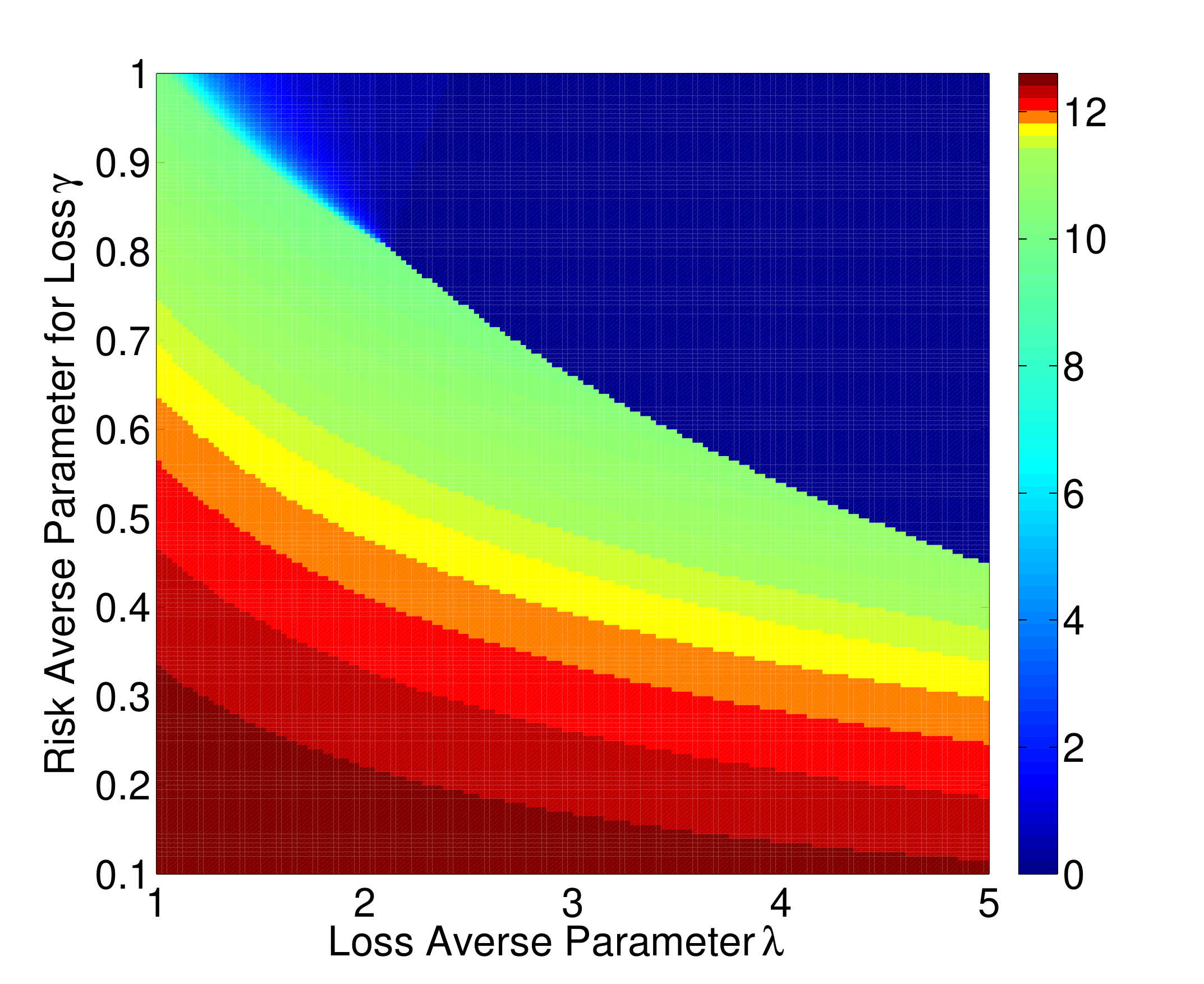}}
  \vspace{-2mm}
  \centerline{\weipp{(a)} $\beta=0.8$, $R_p=(\pi-c_l)D$}
\end{minipage}
\hfill
\begin{minipage}{0.32\linewidth}
  \centerline{\includegraphics[width=5.5cm]{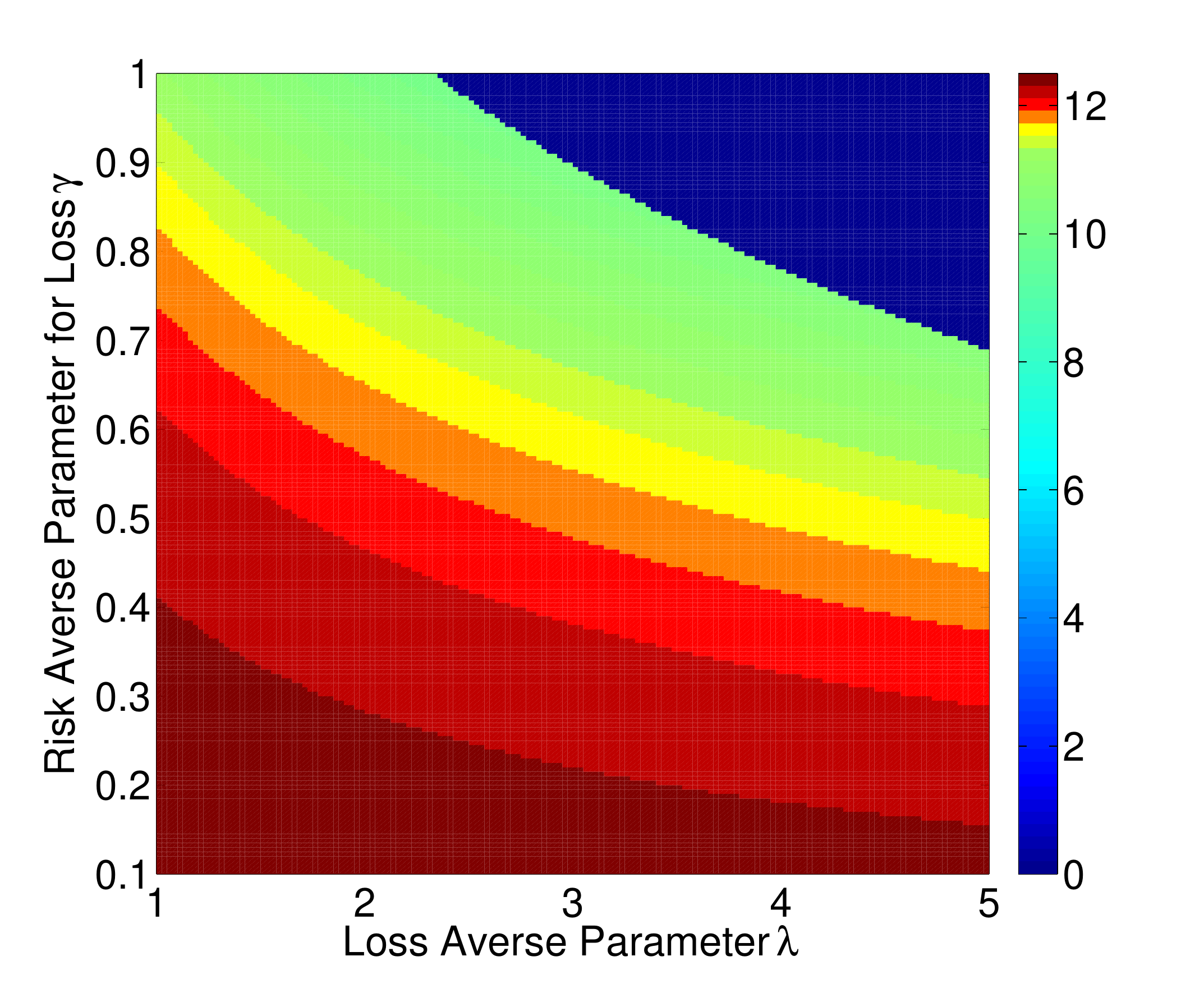}}
    \vspace{-2mm}
  \centerline{\weipp{(b)} $\beta=1$, $R_p=(\pi-c_l)D$}
\end{minipage}
\hfill
\begin{minipage}{0.32\linewidth}
  \centerline{\includegraphics[width=5.5cm]{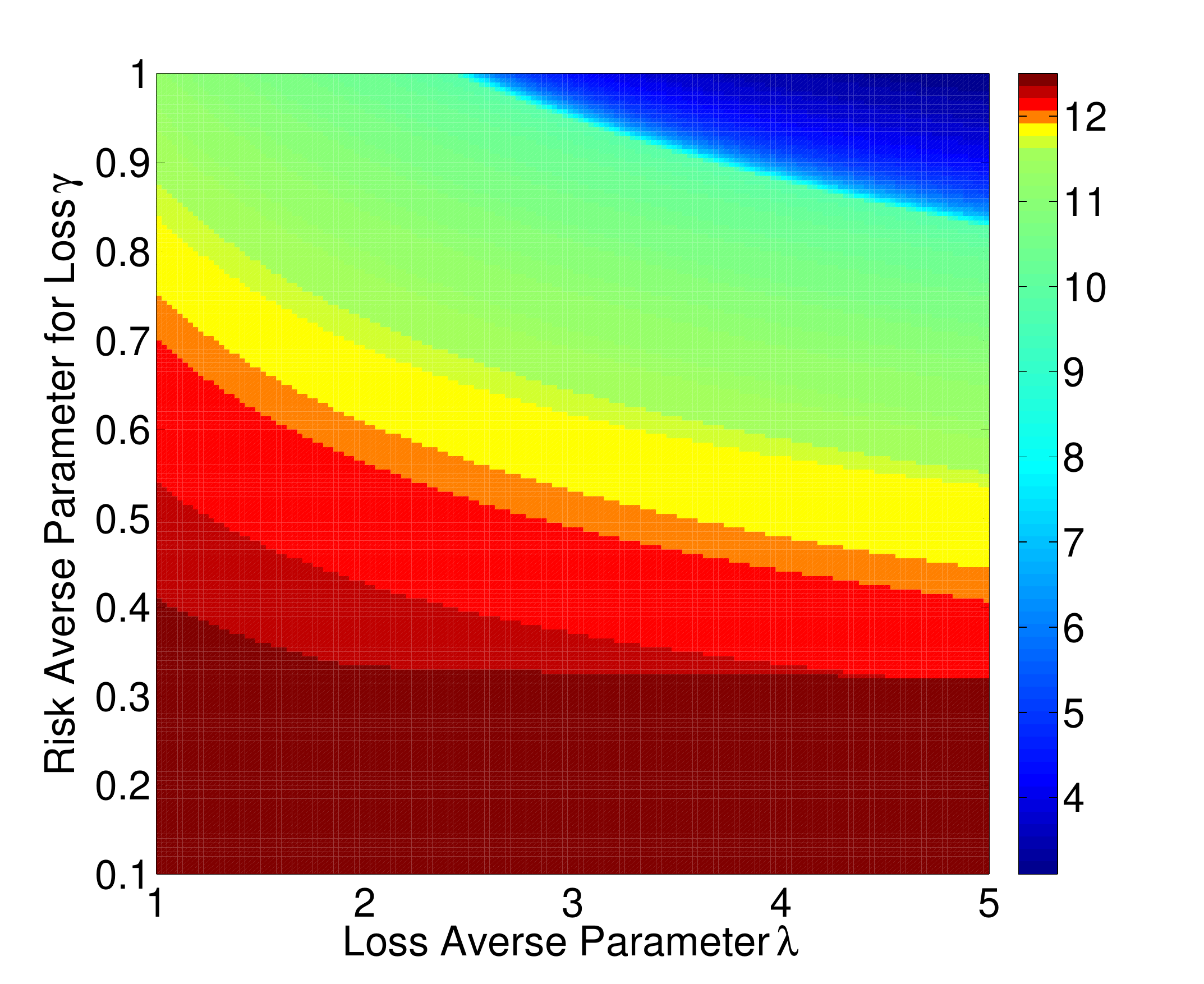}}
    \vspace{-2mm}
  \centerline{(c) $\beta=1$, $R_p=(\pi-c_s)D$}
\end{minipage}
\vfill
	\caption{Optimal sensing decision $B_s^*$ versus $\gamma$ and $\lambda$ for different $\beta$ and $R_p$. Other parameters are $c_l=5$, $c_s=2$, and $D=10$. (The color represents optimal sensing decision $B_s^*$.)}
\label{fig:lambda}
\vspace{-8mm}
\end{figure*}  
\weip{The proof of Theorem 3 is given in Appendix G.} Theorem 3 indicates that $B_s^*$ under $R_p^{\text{high}}$ is always larger than $B_s^*$ under $R_p^{\text{low}}$. This means that an operator with $R_p^{\text{high}}$ is more willing to sense \weipp{compared} to an operator with $R_p^{\text{low}}$. 
When an operator  has a high expectation (due to a high reference point), it is more likely to encounter losses than gains under uncertainty. \weip{Since the operator's valuation function $v(x-R_p)$ is convex
in the loss region, it will sense more in order to gain more in the case of high sensing realization (i.e., $\alpha=1$).} In contrast, when an operator has a low expectation (due to a low reference point), it is more likely to encounter gains than losses. Since the operator's valuation function $v(x-R_p)$ is concave in terms of gains, it will sense less, in order to avoid the risk of low sensing realization (i.e., $\alpha=0$). \revo{To summarize, an operator who expects a higher profit is more risk-seeking, and an operator who expects a lower profit is more risk-averse. }

\section{Performance Evaluations} \label{sec:pe} 

  In this section, we illustrate the operator's optimal sensing decision and the corresponding expected profit  under different system parameters. The key insights under the PT modeling include: \rev{(a) A \emph{risk-averse} operator will sense less and lease more, which leads to a smaller profit with a lower risk of loss; while a \emph{risk-seeking} operator will sense more and lease less, which leads to a larger profit with a higher risk of loss. (b) \emph{Risk preference} changes with the probability of high sensing realization. When the probability of high sensing realization changes from very high to very low, the operator changes from risk-averse to risk-seeking. (c) Both \emph{risk-averse} and \emph{risk-seeking} operators face a tradeoff between satisfying their risk preferences and maximizing expected profit. A \emph{risk-averse} operator achieves a better tradeoff in a high sensing cost scenario than in a low sensing cost scenario, while a \emph{risk-seeking} operator achieves a better tradeoff in a low sensing cost scenario than in a high sensing cost scenario.}
  
\subsection{Evaluation of the Optimal Sensing Decision}
We first evaluate the impact of the three characteristics of PT on the operator's optimal decision. 
\subsubsection{Impact of s-shaped Value Function}  
First, we illustrate the operator's optimal sensing decision under different parameters of s-shaped value function (i.e., $\gamma$, $\beta$, and $\lambda$), assuming reference point $R_p=(\pi-c_l)D$ and linear probability distortion (i.e., $\mu=1$). We compare the optimal sensing decision and the corresponding expected profit with the EUT benchmark, where $\lambda=\beta=\gamma=1$. 

First, in Fig. \ref{fig:lambda}(a) and Fig. \ref{fig:lambda}(b), we study the optimal sensing decision $B_s^*$ against the loss penalty parameter $\lambda$ and the risk averse parameter for loss $\gamma$. To illustrate the impact of $\beta$, we set $\beta=0.8$ in Fig. \ref{fig:lambda}(a) and $\beta = 1$ in Fig. \ref{fig:lambda}(b). 
  The other system parameters are fixed at $c_l=5$, $c_s=2$, and $D=10$. 
  We observe the behaviors of both risk-averse and risk-seeking operators in Fig. \ref{fig:lambda}(a) and Fig. \ref{fig:lambda}(b). The upper right parts of the figures \weipp{correspond} to the risk-averse operators, and the lower left parts of the figures are risk-seeking operators. The EUT benchmark corresponds to the case of $\lambda = \beta = \gamma = 1$ (i.e., upper left corners of Fig. 3(b) and Fig. 3(c)). We observe that risk-averse operators sense less and risk-seeking operators sense more.
  
  \textbf{Impact of $\gamma$ on $B_s^*$:} We can see that for fixed $\beta$ and $\lambda$, the operator senses more when $\gamma$ decreases (see y-axis in Fig. \ref{fig:lambda}(a) and Fig. \ref{fig:lambda}(b)), as stated in Corollary 1. \revo{The intuition is that when $\gamma$ decreases from 1 (hence the operator is more risk-seeking than an EUT operator), the operator experiences less marginal disutility from loss. In order to win a potentially large gain, the operator is more willing to take risk and sense more. An example to illustrate this impact is that a gambler, who has already lost a lot, cares less of losing an additional \$$1$ than a gambler who just starts to gamble. }

\textbf{Impact of $\lambda$ on $B_s^*$:} As stated in Corollary 1, for fixed $\gamma$ and $\beta$, the operator senses less when $\lambda$ increases (see the x-axis in Fig. \ref{fig:lambda}(a) and Fig. \ref{fig:lambda}(b)). The intuition is that when $\lambda$ is larger, the penalty of loss to the operator is larger (hence the operator is more risk-averse than an EUT operator). In order to avoid a potential loss, the operator will sense less.

\textbf{Impact of $\beta$ on $B_s^*$:} By comparing Fig. \ref{fig:lambda}(a) and Fig. \ref{fig:lambda}(b), we can observe that the operator senses more when $\beta$ increases for fixed $\lambda$ and $\gamma$, which verifies Corollary 2.  \revo{The intuition is that when $\beta$ decreases from 1 (hence the operator is more risk-averse than an EUT operator), the operator experiences less marginal utility from the same gain. Hence, the operator will sense less to achieve a certain gain, rather than taking risk for a very large gain. An example to illustrate this impact is that a rich man is less willing to earn an additional \$$1$ than a poor man if doing so requires a fixed amount of effort. }
  
\subsubsection{Impact of Reference Point}  
In Fig. \ref{fig:lambda}(b), we have considered the reference point of $R_p=D(\pi-c_l)$. Next, we further illustrate the operator's optimal sensing decision under another high reference point of $R_p=D(\pi-c_s)$. 
\weip{As stated in Section V-B, whether an outcome is considered a loss and gain will significantly affect the operator's subjective valuation of the outcome, hence will affect its sensing decision.}

In Fig. \ref{fig:lambda}(c), we consider the case of $R_p=D(\pi-c_s)$, which reflects the operator's expectation of realizing all of the sensing spectrum $D$. Hence the same outcome is more likely to be considered as a loss under $R_p = D(\pi-c_s)$ than under $R_p=D(\pi-c_l)$. We plot the optimal sensing decision of the risk-seeking and risk-averse operators for different values of $\lambda$ and $\gamma$. The other system parameters are the same as those in Fig. \ref{fig:lambda}(b). By comparing Fig. \ref{fig:lambda}(b) with Fig. \ref{fig:lambda}(c), we observe that the operator with a low reference point $R_p=D(\pi-c_l)$ senses less. The intuition is that when an operator has a low reference point, it has a low expectation, so it is more likely to encounter gains than losses. Due to the concavity of its valuation function $v(x-R_p)$ in gain, the operator will become more risk-averse and will sense less to avoid the risk of low sensing realization.

		\begin{figure}[t]
	\centering
		\includegraphics[width=0.38\textwidth]{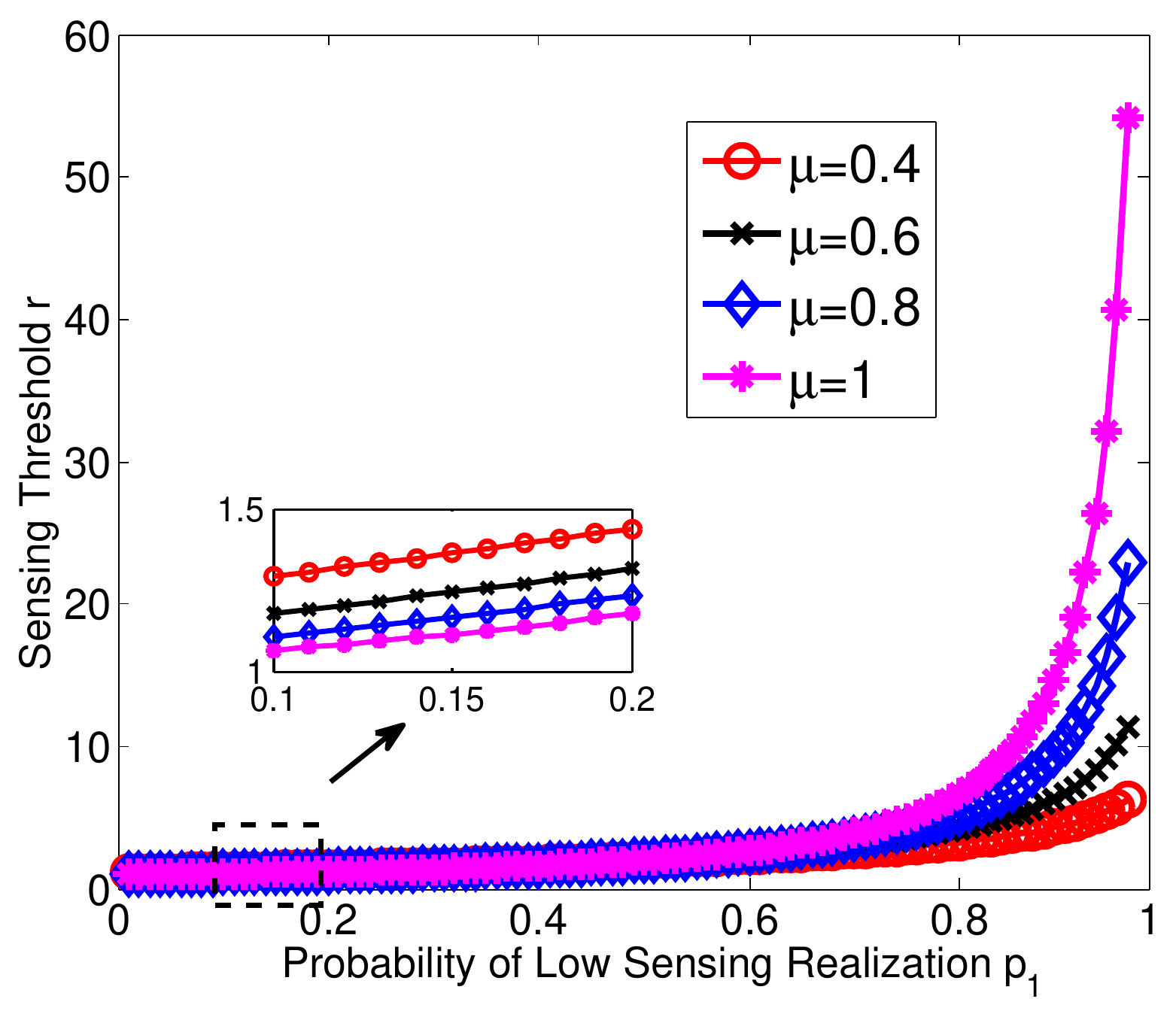}
		\vspace{-4mm}
	\caption{Sensing threshold $r$ versus probability of low sensing realization $p_1$ with different probability distortion parameter $\mu$ for $\beta = \gamma = \lambda = 1$, and $R_p = D(\pi -  c_l)$.}
	\label{fig:mu}	
	\vspace{-6mm}
\end{figure}


\subsubsection{Impact of Probability Distortion}  
Then, we illustrate the operator's optimal sensing decision under different probability distortion parameters $\mu$ for the case of binary sensing outcomes (i.e., $I=2$). We compare the result with the non-distorted benchmark, where $\mu=1$. \weipp{The PT operator is \emph{risk-seeking} when the probability of low sensing realization $p_1$ is high, and it is \emph{risk-averse} when $p_1$ is low.}

We notice that the operator's decisions can be characterized by a threshold $r$ related to the leasing and sensing cost ratio $\frac{c_l}{c_s}$. When $\frac{c_l}{c_s} \geq r$, meaning that leasing is expensive, the operator will choose to sense for all the demand $D$. Otherwise, when $\frac{c_l}{c_s} < r$, the operator will sense less than the demand $D$ and lease for part of the demand\footnote{\weipp{Notice that the threshold $r$ is different under different scenarios. For example, with $R_p=D(\pi-c_l)$ for binary sensing outcomes, we have $r = \left[w\left(p_1\right)/w\left(p_2\right)\right]^{\frac{1}{\beta}}+1$ from Theorem 2. On the other hand, with $R_p=D(\pi-c_s)$ for binary sensing outcomes, we have $r=1+w(p_1)/w(p_2)$ from Theorem 3.}}. \wei{Hence, a larger $r$ means that the operator is less willing to sense.}

In Fig. \ref{fig:mu}, we plot the sensing threshold $r$ against $p_1$ for different values of $\mu$, where we assume $\beta=\gamma=\lambda=1$, and $R_p=D(\pi-c_l)$. We can see that the threshold $r$ decreases in $\mu$ when $p_1<0.5$, and increases in $\mu$ when $p_1>0.5$. As a smaller $\mu$ means that the operator will overweigh the low probability more, it becomes more risk-averse when $p_1$ is small. Similarly, since a smaller $\mu$ means that the operator will underweigh the high probability more, it is more risk-seeking when $p_1$ is large. 

\subsection{Expected Profit and Risk Preference Tradeoff}  
\begin{figure}[t]
\begin{minipage}{0.47\linewidth}
  \centerline{\includegraphics[width=4.7cm]{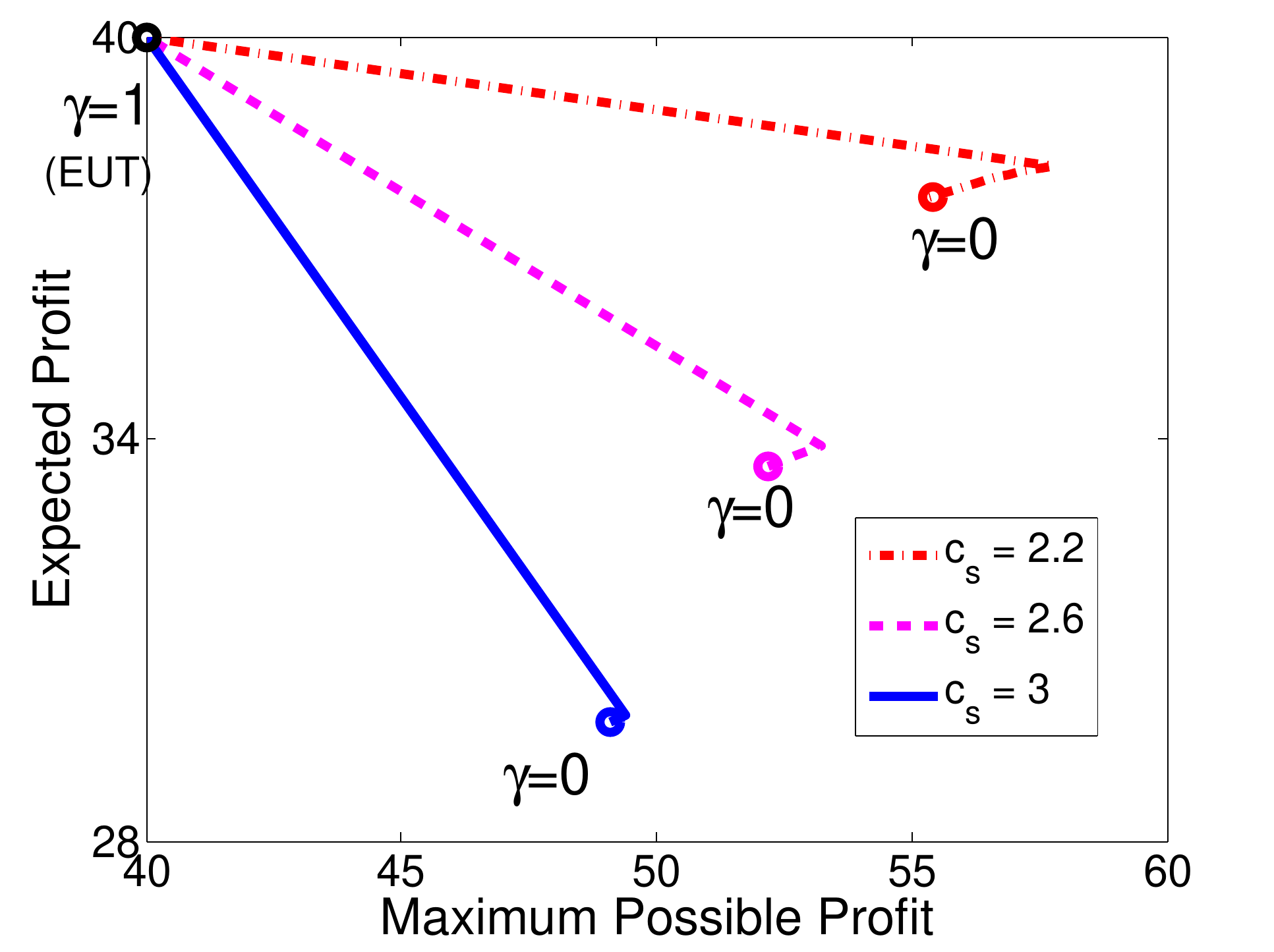}}
\end{minipage}
\hfill
\begin{minipage}{0.52\linewidth}
  \centerline{\includegraphics[width=4.7cm]{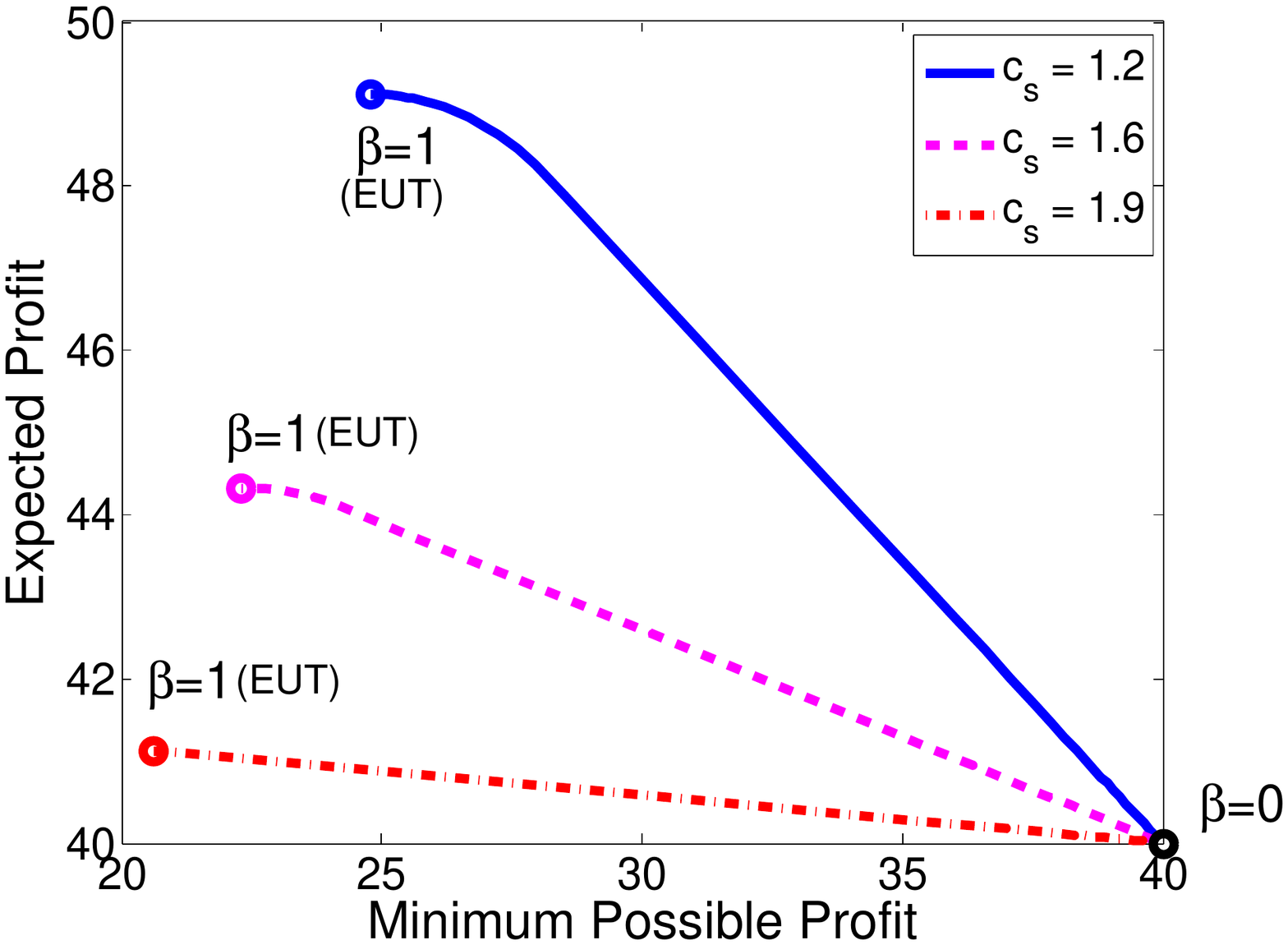}}
\end{minipage}
\vfill
\vspace{-1mm}
\caption{\revo{(a) Tradeoff between optimal expected profit and maximum possible profit for $\beta = 1$ under different $c_s$ and $\gamma$. 
(b) Tradeoff between optimal expected profit and minimum possible profit for $\gamma = 1$ under different $c_s$ and $\beta$. 
Other parameters are $R_p=(\pi-c_l)D$, $c_l=4$, $\pi=8$, and $D=10$.}}
\label{fig:lambdaer}\vspace{-5mm}
\end{figure}

\revoo{We then evaluate the tradeoff between the expected profit and risk preference of an operator.
A risk-seeking operator is aggressive and mainly interested in earning a high maximum profit, while a risk-averse operator is conservative and mainly interested in guaranteeing a high minimum profit.
Given the system parameters (i.e., $c_s$, $c_l$, $\pi$, and $D$) and risk perference parameters (i.e., $\lambda$, $\beta$, and $\gamma$), we let $B_s^*$ and $B_l^*$ be the optimal sensing and leasing \weipp{decisions} discussed in Section IV. The optimal expected profit $\mathbb{E}_{\alpha}[R(B_s^*,B_l^*,\alpha)]$ is the averaged profit over different sensing realizations. The maximum possible profit $R(B_s^*,B_l^*,1)$ is the profit with $\alpha=1$, and the minimum possible profit $R(B_s^*,B_l^*,0)$ is the profit with $\alpha=0$.

\textbf{Tradeoff of a risk-seeking operator:} In Fig. \ref{fig:lambdaer}(a), we plot the tradeoff between the expected profit and the maximum possible profit for $\beta = 1$ under different $c_s$ and $\gamma$.
Since an EUT operator ($\gamma=1$) makes decision only by maximizing expected profit, we can see from Fig. \ref{fig:lambdaer}(a) that it can achieve the highest expected profit. On the other hand, a PT operator makes decision by taking into account both the expected profit and its risk preference. More specifically, although a risk-seeking PT operator (i.e., $\gamma<1$) achieves a lower expected profit comparing to an EUT operator, it can earn a higher maximum possible profit than an EUT operator. 
\weipp{Notice in Fig. 5(a), when the operator is very risk-seeking ($\gamma\rightarrow 0$), the expected profit and maximum possible profit both decrease. This is because the operator can achieve the maximum possible profit when the sensed spectrum is fully realized. However, when the sensing decision $B_s^*$ is larger than demand D (hence the maximum realized spectrum is larger than D), being more risk-seeking (which leads to a larger sensing decision $B_s^*$) will not lead to a larger maximum possible profit, but will only lead to a larger probability of achieving that maximum possible profit. This explains the ``bending" in the figure. }
From Fig. \ref{fig:lambdaer}(a), we also observe that the tradeoff varies with sensing cost $c_s$. \weipp{We can see that under the three cases of $c_s$, an EUT operator has the same expected profit and maximum possible profit, because it has the same optimal sensing decision ($B_s^*=0$). However, a risk-seeking operator will have a smaller loss in expected profit but a larger gain in maximum possible profit than an EUT operator ($\gamma = 1$) when the sensing cost $c_s$ decreases.} In other words, a risk-seeking operator achieves a better tradeoff when the sensing cost decreases.

\textbf{Tradeoff of a risk-averse operator:} In Fig. \ref{fig:lambdaer}(b), we plot the tradeoff between optimal expected profit and minimum possible profit for $\gamma = 1$ under different $c_s$ and $\beta$. 
We can see from Fig. \ref{fig:lambdaer}(b) that a risk-averse operator ($\beta<1$) achieves a lower expected profit comparing to an EUT operator ($\beta=1$), but guarantees a higher minimum possible profit than an EUT operator. \weipp{For example, an extremely risk-averse operator ($\beta\rightarrow0$) has a similar expected profit and minimum possible profit under any sensing cost $c_s$, because its optimal sensing decision $B_s^*$ is always close to zero.} We can also observe that \weipp{a} risk-averse operator will have a smaller loss in expected profit but a larger gain in minimum possible profit than an EUT operator when the sensing cost $c_s$ increases. In other words, a risk-averse operator achieves a better tradeoff when the sensing cost increases.}

\section{Conclusions and Future Work} \label{sec:concl}

In this paper, we considered a spectrum investment problem with sensing uncertainty, where an operator decides its spectrum sensing and leasing decisions by considering both expected profit and its risk preference based on prospect theory. 
This is the first paper that studied the optimal decisions based on all three characteristics of prospect theory in the wireless communication literature, and compared and contrasted these decisions with those under the more widely used expected utility theory. 
\wei{Our results suggested that a risk-averse operator can achieve a large expected profit while guaranteeing a satisfactory level of minimum possible profit when the sensing cost is high. On the other hand, a risk-seeking operator can achieve both a large expected profit and maximum possible profit when the sensing cost is low.}

	This study demonstrated that a more realistic modeling based on prospect theory is important in understanding the operator's decisions in the wireless industry. On the other hand, this study is only a small first step, as we have only considered the operator's decision in a single time slot. Regarding the future work, we will consider a more general problem with decisions to be made in multiple time slots. In such a model, the operator's reference point may change over time, and the study of dynamic reference point is a recent active research field in prospect theory \cite{jin_bp08, yin2004markowitz}. \weipp{We will also conduct a survey to evaluate different people's risk preferences.}

\bibliographystyle{IEEEtran}
\bibliography{IEEEabrv,mybibfile}

\appendix

\subsection{Proof of Theorem \ref{thm:pt_table}} \label{app:pt_table}
In this proof, we divide the feasible range of $B_s$ into three intervals, $\left[0,\frac{D}{\alpha_I}\right]$, $\left[\frac{D}{\alpha_I},\frac{D c_l}{c_s}\right]$, and $\left(\frac{D c_l}{c_s},\infty\right)$, and analyze the optimal decision $B_s^*$ in each interval.
\weipp{By this division, in the interval $\left[0,\frac{D}{\alpha_I}\right]$, $B_s\alpha_i \leq D$ for all $i \in \mathcal{I}$, so that we do not need to consider the possibility of sensing realization exceeds the demand. In the interval $(\frac{Dc_l}{c_s},\infty)$, we have $Dc_l < B_sc_s$, which means that the total cost of sensing $B_sc_s$ will be larger than the cost of leasing only $Dc_l$. Hence the optimal solution will not be in this range. The above reason leads to the division of three intervals.}

We will use the following notations in the proof:
\begin{align}
M_{I} \triangleq \left[\frac{\sum\limits_{i={\hat{\imath}}+1}^I \left(c_l\alpha_i\!-\!c_s\right)^{\beta}w\!\left(p_i\right)\beta}{\sum\limits_{i=1}^{\hat{\imath}} \lambda \gamma\left(c_s\!-\!c_l\alpha_i\right)^{\gamma}\left(\frac{1}{\alpha_{I}}\right)^{\gamma-\beta}w\!\left(p_i\right)}\right]^{\frac{1}{\gamma-\beta}},
\end{align}
\begin{align}
M_j &\!\triangleq\!\left[\frac{\sum\limits_{i={\hat{\imath}}+1}^j\!\!\!\!\! \left(\!c_l\alpha_i\!-\!c_s\!\right)^{\beta}\!w\!\left(\!p_i\!\right)\!\beta\!-\!\beta c_s\!\left(\!c_l\alpha_{j}\!-\!c_s\!\right)^{\beta-1}\!\!\!\!\!\sum\limits_{i=j+1}^I \!\!\!w\!\left(\!p_i\!\right)}{\sum\limits_{i=1}^{\hat{\imath}} \lambda \gamma\left(c_s-c_l\alpha_i\right)^{\gamma}\left(\frac{1}{\alpha_{j}}\right)^{\gamma-\beta}w\!\left(p_i\right)}\right]^{\frac{1}{\gamma-\beta}}\!\!\!\!\!\!\!\!,
\notag\\
&~~j={\hat{\imath}}\!+\!1,{\hat{\imath}}\!+\!2,...,I\!-\!1,
\end{align}
and
\begin{align}
H_j\! &\triangleq \!\!\left[\frac{\sum\limits_{i={\hat{\imath}}\!+\!1}^j \!\!\!\left(\!c_l\alpha_i\!\!-\!\!c_s\!\right)^{\beta}\!w\!\left(\!p_i\!\right)\!\beta\!\!-\!\!\beta c_s\!\left(\!c_l\alpha_{j\!+\!1}\!\!-\!\!c_s\!\right)^{\beta\!-\!1}\!\!\!\!\sum\limits_{i=j\!+\!1}^I \!\!w\!\left(p_i\right)}{\sum\limits_{i=1}^{\hat{\imath}} \lambda \gamma\left(c_s-c_l\alpha_i\right)^{\gamma}\left(\frac{1}{\alpha_{j+1}}\right)^{\gamma\!-\!\beta}\!w\!\left(\!p_i\!\right)}\right]^{\frac{1}{\gamma-\beta}}\!\!\!\!\!\!\!\!,
\notag\\&
~~ j={\hat{\imath}}\!+\!1,{\hat{\imath}}\!+\!2,...,I\!-\!1.
\end{align}

Next, we analyze the optimal decision $B_s^*$ within each interval.

\textbf{Case I:} $B_s\in\left[0,\frac{D}{\alpha_I}\right]$. We first compute the optimal decision $B_s^*$ in this interval using the first order condition. In this case, $B_s\alpha_i\leq D$  for all $i \in \mathcal{I}$, so the optimal leasing decision $B_l^*=D-B_s\alpha_i\geq0$. From (7), the revenue is 
\begin{equation}
\vspace{-1mm}
R\left(B_s,B_l^*,\alpha_i\right)=\left(\pi-c_l\right)D-B_sc_s+B_sc_l\alpha_i.
\vspace{-1mm}
\end{equation}

Since $\alpha_i$ follows a discrete distribution in $\left[0,1\right]$, we can plug (18) into (9), and get (19) by taking the proper expectation.
\begin{align}\label{equ:case1}
U&\!\left(\!B_s\!\right)\!=\!\!\!\!\sum_{i={\hat{\imath}}\!+\!1}^I\!\!\!\left(B_sc_l\alpha_i\!-\!B_sc_s\right)^{\beta}\!w\!\left(\!p_i\!\right)\!\!-\!\!\!\sum_{i=1}^{\hat{\imath}}\!\lambda\! \left(B_sc_s\!\!-\!\!B_sc_l\alpha_i\right)^{\gamma}\!w\!\left(\!p_i\!\right)\notag\\
=\!\!\!&\sum_{i={\hat{\imath}}\!+\!1}^I\!\!\! \left(c_l\alpha_i\!-\!c_s\right)^{\beta}w\!\left(\!p_i\!\right)\!B_s^{\beta}\!\!-\!\sum_{i=1}^{\hat{\imath}} \!\lambda\! \left(\!c_s-\!c_l\alpha_i\right)^{\gamma}\!w\!\left(\!p_i\!\right)\!B_s^{\gamma}.
\end{align}
We consider the first order derivative of \eqref{equ:case1} with respect to $B_s$:
\begin{align}
U'\!\!&\left(\!B_s\!\right)\!\!=\!\!\!\!\!\sum_{i={\hat{\imath}}\!+\!1}^I\!\!\!\! \left(c_l\alpha_i\!-\!c_s\right)^{\beta}\!\!w\!\left(\!p_i\!\right)\!\beta B_s^{\beta\!-\!1}\!\!-\!\!\!\sum_{i=1}^{\hat{\imath}} \!\!\lambda \gamma\!\left(c_s\!\!-\!c_l\alpha_i\right)^{\gamma}\!\!w\!\left(\!p_i\!\right)\!B_s^{\gamma\!-\!1}\notag\\
=&B_s^{\beta\!-\!1}\!\!\left[\sum_{i={\hat{\imath}}\!+\!1}^I\!\!\! \left(c_l\alpha_i\!-\!c_s\right)^{\beta}\!\!w\!\left(\!p_i\!\right)\!\beta\!-\!\!\!\sum_{i=1}^{\hat{\imath}}\! \lambda \gamma\!\left(c_s\!-\!c_l\alpha_i\right)^{\gamma}\!\!w\!\left(\!p_i\!\right)\!B_s^{\gamma\!-\!\beta}\right]\!.
\end{align}
We can obtain
\begin{align} \label{equ:hfunc}
&U'\left(B_s\right)>0 \Leftrightarrow \sum_{i=1}^{\hat{\imath}} -\lambda \gamma\left(-c_l\alpha_i+c_s\right)^{\gamma}w\!\left(p_i\right)B_s^{\gamma-\beta}\notag\\
&~~\quad\quad\quad\quad\quad\quad+\sum_{i={\hat{\imath}}\!+\!1}^I \left(c_l\alpha_i-c_s\right)^{\beta}w\!\left(p_i\right)\beta>0\notag\\
&\Leftrightarrow B_s<\left[\frac{\sum_{i={\hat{\imath}}+1}^I \left(c_l\alpha_i\!-\!c_s\right)^{\beta}w\!\left(p_i\right)\beta}{\sum_{i=1}^{\hat{\imath}} \lambda \gamma\left(c_s\!-\!c_l\alpha_i\right)^{\gamma}w\!\left(p_i\right)}\right]^{\frac{1}{\gamma-\beta}}\!\!=\frac{M_I}{\alpha_I},
\end{align}
where $M_I$ is defined in (15). Since the second order derivative $U''\!\left(B_s\right)\!<\!0$, we know that $U'\!\left(B_s\right)$ in (20) decreases in $B_s$. Hence, we only need to compare the right boundary $\frac{D}{\alpha_I}$ and the critical point $\frac{M_I}{\alpha_I}$. From \eqref{equ:hfunc}, we have
\begin{equation} 
     B_s^*=\left\{
    \begin{aligned}
    &\frac{M_I}{\alpha_I},  &   &\text{ if }\frac{M_I}{\alpha_I}<\frac{D}{\alpha_I},\\
    &\frac{D}{\alpha_I}, &   &\text{ if }\frac{M_I}{\alpha_I}\geq\frac{D}{\alpha_I}.\\
    \end{aligned}
    \right.\\
\end{equation}
Intuitively, when the leasing cost $c_l$ is high (and hence every $H_j$ in (17) is large), the utility increases when the operator senses more (and leases less). 

\textbf{Case II}: $B_s\in\left[\frac{D}{\alpha_I},\frac{Dc_l}{c_s}\right]$.
We compute the optimal decision $B_s^*$ in this interval by capturing the unimodal structure. In this case, profit $R\left(B_s,B_l^*,\alpha\right)$ can be represented as a piecewise function as follows.
\begin{equation} 
     R\left(B_s,B_l^*,\alpha\right)=\left\{
    \begin{aligned}
    &\pi D-B_sc_s,  &   &\text{ if }\alpha>\frac{D}{B_s},\\
    &B_sc_l\alpha-B_sc_s, &   &\text{ if }0<\alpha<\frac{D}{B_s}.\\
    \end{aligned}
    \right.\\
\end{equation}
We substitute (23) into (9), and the utility function becomes
\begin{align}\label{equ:fff}
U&\!\left(B_s\right)\!=\!\sum_{i=1}^{\hat{\imath}}\! -\lambda \!\left[B_s\!\left(c_s\!-\!c_l\alpha_i\right)\right]^{\gamma}\!w\!\left(p_i\right)\!\notag\\
+&\!\!\sum_{i={\hat{\imath}}+1}^j\! \!\left[B_s\!\left(c_l\alpha_i\!\!-\!\!c_s\right)\right]^{\beta}\!w\!\left(p_i\right)\!+\!\!\!\sum_{i=j+1}^I\!\!\left(Dc_l\!\!-\!\!B_sc_s\right)^{\beta}\!w\!\left(p_i\right)\!.
\end{align}

The value of $j$ in (24) varies when $B_s$ belongs to different sub-intervals. The utility function $U\left(B_s\right)$ in (24) is a continuous function in the whole interval, and is differentiable in each of the $I-{\hat{\imath}}$ sub-intervals, $\frac{D}{\alpha_{j+1}}\leq B_s\leq\frac{D}{\alpha_{j}}$, $j={\hat{\imath}}+1$, ..., $I-1$, and $\frac{D}{\alpha_{{\hat{\imath}}+1}}\leq B_s\leq\frac{Dc_l}{c_s}$. Although the utility function is not globally differentiable, we can evaluate the derivative of each sub-interval, and find the optimal point for each of the $I-{\hat{\imath}}$ sub-intervals. Then we can find the optimal solution of the whole interval by comparing the optimal points in the $I-{\hat{\imath}}$ sub-intervals.

We conclude the optimal sensing decision $B_{s,j}^*$ in the interval $\frac{D}{\alpha_{j+1}}\leq B_s\leq \frac{D}{\alpha_{j}}$, $j={\hat{\imath}}+1$, ..., $I-1$ in Proposition 1.

\begin{pps} \label{thm:onemax}
The maximum utility value $U\left(B_{s,j}^*\right)$ in each sub-interval among $\frac{D}{\alpha_{j+1}}\leq B_s\leq \frac{D}{\alpha_{j}}$, $j={\hat{\imath}}+1$, ..., $I-1$ is achieved at
\begin{equation} 
     B_{s,j}^*=\left\{
    \begin{aligned}
    &\frac{D}{\alpha_{j}}, &   &\text{ if }D\leq M_{j},\\
    &g_{j}^{-1}\left(0\right), &   &\text{ if }M_j< D<H_{j},\\
        &\frac{D}{\alpha_{j+1}},  &   &\text{ if }H_j\leq D<M_{j+1}.\\
    \end{aligned}
    \right.\\
\end{equation}
\end{pps}
\begin{proof}
We prove Proposition 1 by capturing the unimodal structure of (24) in each sub-interval $\frac{D}{\alpha_{j+1}}\leq B_s\leq \frac{D}{\alpha_{j}}$. For a unimodal problem, the optimal point is either at the unique local maximum point or the boundaries.
For $\frac{D}{\alpha_{j+1}}\leq B_s\leq \frac{D}{\alpha_{j}}$, we consider the first order derivative of (24) with respect to $B_s$:
\begin{align}
&U'\!\left(\!B_s\right)\!\!=\!\!\sum\limits_{i=1}^{\hat{\imath}}\!\!-\!\lambda \left(c_s\!-\!c_l\alpha_i\right)^{\gamma}w\!\left(p_i\right)\!\gamma B_s^{\gamma-1}\!\!\notag\\
&~~+\!\!\!\sum\limits_{i={\hat{\imath}}\!+\!1}^j\!\!\beta\! \left(c_l\alpha_i\!-\!c_s\right)^{\beta}\!\!w\!\left(\!p_i\!\right)\!B_s^{\beta\!-\!1}\!\!-\!c_s\!\!\!\sum\limits_{i=j\!+\!1}^I\!\!\beta\!\left(Dc_l\!-\!B_sc_s\right)^{\beta\!-\!1}\!\!w\!\left(\!p_i\!\right)\!\notag\\
&=\sum\limits_{i={\hat{\imath}}+1}^j\!\beta\left(c_l\alpha_i\!-\!c_s\right)^{\beta}w\!\left(p_i\right)\!B_s^{\beta-1}\!\!\Bigg[\frac{\sum\limits_{i=1}^{\hat{\imath}}-\lambda\left(c_s\!-\!c_l\alpha_i\right)^\gamma B_s^{\gamma\!-\!\beta}}{\sum\limits_{i={\hat{\imath}}+1}^j \left(c_l\alpha_i-c_s\right)^{\beta}w\!\left(p_i\right)\!\beta}\!\notag\\
&~~~~+\!1\!-\!\frac{c_s\!\!\sum\limits_{i=j+1}^I \!\!\left(\frac{Dc_l}{B_s}\!-\!c_s\right)^{\beta-1}\!w\!\left(p_i\right)\!}{\sum\limits_{i={\hat{\imath}}+1}^j \!\!\left(c_l\alpha_i\!-\!c_s\right)^{\beta}w\!\left(p_i\right)\!}\Bigg],
\end{align}
where we define
\begin{align}
g_j&\left(B_s\right) \triangleq \frac{\sum\limits_{i=1}^{\hat{\imath}}\!-\lambda\left(c_s\!-\!c_l\alpha_i\right)^\gamma B_s^{\gamma-\beta}}{\sum\limits_{i={\hat{\imath}}+1}^j\! \left(c_l\alpha_i\!-\!c_s\right)^{\beta}w\!\left(p_i\right)\beta}\!+\!1\!\notag\\
&-\!\frac{c_s\left(\frac{Dc_l}{B_s}\!-\!c_s\right)^{\beta-1}\!\!\!\sum\limits_{i=j+1}^I\! w\!\left(p_i\right)}{\sum\limits_{i={\hat{\imath}}+1}^j \!\left(c_l\alpha_i\!-\!c_s\right)^{\beta}w\!\left(p_i\right)}, ~ j={\hat{\imath}}\!+\!1,{\hat{\imath}}\!+\!2,...,I\!-\!1.
\end{align}

With (27), we can rewrite $U'\left(B_s\right)$ as 
\begin{equation}
    U'\left(B_s\right)=\sum_{i={\hat{\imath}}+1}^j \beta\left(c_l\alpha_i-c_s\right)^{\beta}w\!\left(p_i\right)B_s^{\beta-1}g_j\left(B_s\right),
\end{equation}
where $\sum_{i={\hat{\imath}}+1}^j \beta\left(c_l\alpha_i-c_s\right)^{\beta}w\!\left(p_i\right)B_s^{\beta-1}>0$. The function $g_j\left(B_s\right)$ in (27) is a strictly decreasing function of $B_s$, which means the first order derivative $U'\left(B_s\right)$ in (28) will only be zero at most once, thus there is at most one local maximum point\footnote{If this point is local minimum, then there is no local maximum point, and the maximum point in this sub-interval is at the boundaries.}. 

We then consider the two boundary points: $B_s=\frac{D}{\alpha_{j+1}}$ and $B_s=\frac{D}{\alpha_j}$.

$1)$ When $B_s=\frac{D}{\alpha_{j+1}}$, we have:
\begin{align}
 g_j\!\left(\!\frac{D}{\alpha_{j+1}}\!\right)\!=& \frac{\sum_{i=1}^{\hat{\imath}}\!-\!\lambda\left(-c_l\alpha_i\!+\!c_s\right)^\gamma \!\left(\frac{D}{\alpha_{j+1}}\right)^{\gamma-\beta}\!\!\gamma w\!\left(p_i\right)}{\sum_{i={\hat{\imath}}+1}^j \left(c_l\alpha_i\!-\!c_s\right)^{\beta}w\!\left(p_i\right)\beta}\!+\!1\!\notag\\ &-\!\frac{ c_s\sum_{i=j+1}^I \!\left(c_l\alpha_{j+1}\!-\!c_s\right)^{\beta-1}w\!\left(p_i\right)}{\sum_{i={\hat{\imath}}+1}^j \left(c_l\alpha_i\!-\!c_s\right)^{\beta}w\!\left(p_i\right)}.
\end{align}
We can obtain									
						\begin{align}
						 g_j\left(\frac{D}{\alpha_{j+1}}\right)>0 \Leftrightarrow D<H_j.
						\end{align}
			
		$2)$ When $B_s=\frac{D}{\alpha_{j}}$, we have:
						\begin{align}
 g_j\left(\frac{D}{\alpha_{j}}\right) =& \frac{\sum_{i=1}^{\hat{\imath}}-\lambda\left(-c_l\alpha_i+c_s\right)^\gamma \left(\frac{D}{\alpha_{j}}\right)^{\gamma-\beta}\gamma w\!\left(p_i\right)}{\sum_{i={\hat{\imath}}+1}^j \left(c_l\alpha_i-c_s\right)^{\beta}w\!\left(p_i\right)\beta}+1\notag\\
 &-\frac{ c_s\sum_{i=j+1}^I \left(c_l\alpha_{j}\!-\!c_s\right)^{\beta-1}w\!\left(p_i\right)}{\sum_{i={\hat{\imath}}+1}^j \left(c_l\alpha_i-c_s\right)^{\beta}w\!\left(p_i\right)}.
 \vspace{-1mm}
\end{align}
			We can obtain									
						\begin{align}
						 g_j\left(\frac{D}{\alpha_{j}}\right)>0 \Leftrightarrow D<M_j.
						\end{align}		

We can see that $g_j\left(B_s\right)$ and $U'\left(B_s\right)$ in (28) have the same sign within the interval $\left[\frac{D}{\alpha_{j+1}},\frac{D}{\alpha_j}\right]$. From (30), we have that $U'\left(\frac{D}{\alpha_{j+1}}\right)\geq0$ when $D\geq H_j$ and $U'\left(\frac{D}{\alpha_{j+1}}\right)<0$ when $D<H_j$. From (32), we have that $U'\left(\frac{D}{\alpha_{j}}\right)\geq0$ when $D\geq M_j$ and $U'\left(\frac{D}{\alpha_{j}}\right)<0$ when $D<M_j$. Thus $U'\left(B_s\right)$ can be either negative within the entire interval $\left[\frac{D}{\alpha_{j+1}},\frac{D}{\alpha_j}\right]$, or positive within the entire interval, or first positive and then negative within that interval based on the value of demand parameters $H_j$ and $M_j$.

To sum up, the optimal solution $B_{s,j}^*$ in each sub-interval $B_s\in \left[\frac{D}{\alpha_{j+1}},\frac{D}{\alpha_j}\right]$, $j={\hat{\imath}}+1$, ..., $I-1$ depends on the value of demand parameters $H_j$ and $M_j$ as in (25).
\end{proof}

Then we are going to study the cases at the boundaries of the interval $\left[\frac{D}{\alpha_I},\frac{Dc_l}{c_s}\right]$. The case of left boundary ($B_s=\frac{D}{\alpha_I}$) is included in Proposition 1. When $B_s>\frac{D}{\alpha_{{\hat{\imath}}+1}}$, we can show the utility $U\left(B_s\right)$ in (24) is decreasing in $B_s$, hence the optimal $B_s^*\leq\frac{D}{\alpha_{{\hat{\imath}}+1}}$.

We then study the relation between the adjacent sub-intervals. Since (i) $U\left(B_s\right)$ is continuous, (ii) $g_j\left(B_s\right)>g_{j-1}\left(B_s\right)$ for all $j$, and (iii) $g_j\left(B_s\right)$ is decreasing in $B_s$ for all $j$, we know in $\left[\frac{D}{\alpha_I},\frac{Dc_l}{c_s}\right]$,
\begin{equation} 
     B_s^*=\left\{
    \begin{aligned}
    &\frac{D}{\alpha_j},  &   &\text{ if }g_j\left(\frac{D}{\alpha_j}\right)\geq0 \text{ and } g_{j-1}\left(\frac{D}{\alpha_j}\right)\leq0,\\
    &g_j^{-1}\left(0\right), &   &\text{ if }g_j\left(\frac{D}{\alpha_j}\right)<0\text{ and } g_j\left(\frac{D}{\alpha_{j+1}}\right)>0.\\
    \end{aligned}
    \right.\\
\end{equation}

Based on (30), (32) and (33), we can have the following summary.
\begin{equation} 
     B_s^*=\left\{
    \begin{aligned}
    &\frac{D}{\alpha_{j}}, &   &\text{ if } H_j\leq D\leq M_{j+1},\\
    &g_j^{-1}\left(0\right),  &   &\text{ if } M_j<D<H_j,\\    
    \end{aligned}
    \right.\\
\end{equation}
where $j=\hat{\imath},...,I-1$.
						
\textbf{Case III:} $B_s\in\left(\frac{Dc_l}{c_s},\infty\right)$.
In this case, $Dc_l-B_sc_s<0$. From (7) and (11), we know 
\begin{equation}
\vspace{-1mm}
  R\left(B_s,B_l^*, \alpha\right)-R_p= Dc_l-B_sc_s-\max\{D-B_s\alpha,0\}<0,
\end{equation}
which means the total cost of sensing $B_sc_s$ will be larger than the cost of only using spectrum leasing $Dc_l$.  However, the amount of revenue $D\pi$ from sensing or leasing is the same, since the total demand is limited. Hence it is impossible to choose the optimal $B_s^\ast$ in this range to maximize the utility. 

By summarizing the analysis of the above three cases, we obtain Table \ref{table:pt}.

\subsection{Proof of Corollary \ref{lem:lambda}} \label{app:lambda}

From (26), we obtain 
						{\setlength\abovedisplayskip{3pt plus 5pt minus 7pt} 
\setlength\belowdisplayskip{3pt plus 5pt minus 7pt} 		
\begin{align}
h_{\lambda}\left(\lambda\right) \triangleq \frac{\partial U'\left(B_s\right)}{\partial \lambda}=\sum\limits_{i=1}^{\hat{\imath}}-\gamma w\!\left(p_i\right)\left(c_s-c_l\alpha_i\right)^\gamma B_s^{\gamma-1}<0,
\end{align} 
 }
 and
 \begin{align}
h_\gamma\left(\gamma\right) \triangleq& \frac{\partial U'\left(B_s\right)}{\partial \gamma}=\sum\limits_{i=1}^{\hat{\imath}}-\lambda w\!\left(p_i\right)\left(c_s-c_l\alpha_i\right)^\gamma B_s^{\gamma-1}\notag\\ 
&+\sum\limits_{i=1}^{\hat{\imath}}-\lambda w\!\left(p_i\right)\gamma\ln{\left(c_s-c_l\alpha_i\right)}\left(c_s-c_l\alpha_i\right)^\gamma\notag\\
&-B_s^{\gamma-1}\sum\limits_{i=1}^{\hat{\imath}}\lambda w\!\left(p_i\right)\gamma\ln{B_s}\left(c_s\!-\!c_l\alpha_i\right)^\gamma B_s^{\gamma-1}<0.
\end{align} 

If $\lambda_1>\lambda_2$, then $h_{\lambda}\left(\lambda_1\right)<h_{\lambda}\left(\lambda_2\right)$ for every $B_s$, which means $U'\left(B_s\right)$ with $\lambda = \lambda_1$ is less than $U'\left(B_s\right)$ with $\lambda = \lambda_2$ for every $B_s$, and $U'\left(B_s\right)$ with $\lambda = \lambda_1$ will become zero with a smaller $B_s$. Hence $B_s^*$ decreases in $\lambda$.

If $\gamma_1>\gamma_2$, then $h_{\gamma}\left(\gamma_1\right)<h_{\gamma}\left(\gamma_2\right)$ for every $B_s$, which means $U'\left(B_s\right)$ with $\gamma = \gamma_1$ is less than $U'\left(B_s\right)$ with $\gamma = \gamma_2$ for every $B_s$, and $U'\left(B_s\right)$ with $\gamma = \gamma_1$ will be zero with a smaller $B_s$. Hence $B_s^*$ decreases in $\gamma$ when $D>M_I$.

\subsection{Proof of Corollary \ref{lem:beta}} \label{app:beta}

From (21), we obtain
 \begin{align}
h_\beta&\left(\beta\right) \triangleq \frac{\partial U'\!\left(B_s\right)}{\partial \beta}=\!\!\sum\limits_{i={\hat{\imath}}+1}^I \!\!w\!\left(p_i\right)\left(c_l\alpha_i\!-\!c_s\right)^\beta B_s^{\beta-1}\!\!\notag\\
&-\!\!\sum\limits_{i={\hat{\imath}}+1}^I\lambda w\!\left(p_i\right)\beta\ln{\left(c_l\alpha_i\!-\!c_s\right)}\left(c_l\alpha_i\!-\!c_s\right)^\beta\notag\\ 
&-B_s^{\beta-1}\!\!\!\sum\limits_{i={\hat{\imath}}+1}^I\!\!\!\lambda w\!\left(p_i\right)\beta\ln{B_s}\!\left(c_l\alpha_i\!-\!c_s\right)^\beta\!B_s^{\beta-1}\!<\!0.
\end{align}

If $\beta_1>\beta_2$, then $h_\beta\left(\beta_1\right)<h_\beta\left(\beta_2\right)$ for every $B_s$, which means $U'\left(B_s\right)$ with $\beta = \beta_1$ is less than $U'\left(B_s\right)$ with $\beta = \beta_2$ for every $B_s$, and $U'\left(B_s\right)$ with $\beta = \beta_1$ will be zero with a smaller $B_s$. Hence $B_s^*$ decreases in $\beta$.

\subsection{Proof of Corollary \ref{lem:cs}} \label{app:cs}

From (26), we obtain
 \begin{align}
h_{c_s}\left(c_s\right) \triangleq& \frac{\partial U'\left(B_s\right)}{\partial c_s}=\sum\limits_{i=1}^{\hat{\imath}}-\lambda{\gamma}^2 w\!\left(p_i\right)\left(c_s-c_l\alpha_i\right)^{\gamma-1} B_s^{\gamma-1}\notag\\
&+c_s\!\sum\limits_{i=j+1}^I B_s\beta \left(\beta-1\right)w\!\left(p_i\right)\left(Dc_l\!-\!B_sc_s\right)^{\beta-2}\notag\\ &-\!\!\!\sum\limits_{i=j+1}^{I}\!\!\beta\left(Dc_l\!-\!B_sc_s\right)^{\beta\!-\!1} w\!\left(p_i\right)\!\!\notag\\
&-\!\!\!\sum\limits_{i=\hat{\imath}+1}^{j}\!\!{\beta}^2 w\!\left(p_i\right)\left(c_l\alpha_i-\!\!c_s\right)^{\beta-1}B_s^{\beta\!-\!1}\!<\!0.
\end{align}

If $c_s^1>c_s^2$, then $h_{c_s}\left(c_s^1\right)<h_{c_s}\left(c_s^2\right)$ for every $B_s$, which means $U'\left(B_s\right)$ with $c_s = c_s^1$ is less than $U'\left(B_s\right)$ with $c_s = c_s^2$ for every $B_s$, and $U'\left(B_s\right)$ with $c_s = c_s^1$ will be zero with a smaller $B_s$. Hence $B_s^*$ decreases in $c_s$.


\subsection{The EUT Benchmark in Section IV}

For comparison, we also consider the operator's optimal sensing and leasing decisions under the EUT model. As mentioned in Section III, EUT model is a special case of the PT model with $\beta = 1$, $\gamma=1$, $\lambda = 1$, and $\mu = 1$. Under the EUT model, we can obtain the solution of problem \eqref{equ:problem_sense} analytically.

\begin{table*}[t]  
\centering 
\caption{Optimal Sensing Decision and Leasing Decision under EUT}
\begin{tabular}{|c|c|c|}
\hline
\textbf{Condition}&\textbf{Optimal Sensing Decision} $B_s^*$&\textbf{Optimal Leasing Decision} $B_l^*$\\ \hline
$\frac{c_l}{c_s}\leq \left(\sum\limits_{i=1}^I p_i\alpha_i\right)^{-1}$&$B_s^*=0$&$B_l^*=D$\\ \hline
$\left(\sum\limits_{i=1}^{j+1} p_i\alpha_i\right)^{-1}\!<\!\frac{c_l}{c_s}\!<\!\left(\sum\limits_{i=1}^j p_i\alpha_i\right)^{-1}$ for $j = 1,...,I\!-\!1$&$B_s^*=\min\{\frac{D}{\alpha_{j+1}},\frac{Dc_l}{c_s}\}$&$B_l^*\!=\!\max\{0,D\!-\!\alpha \min\{\frac{D}{\alpha_{j+1}},\frac{Dc_l}{c_s}\}\}$\\ \hline
$\frac{c_l}{c_s}\geq\left(\alpha_1p_1\right)^{-1}$&$B_s^*=\frac{D}{\alpha_1}$&$B_l^*=D-\alpha \frac{D}{\alpha_1}$\\\hline
\end{tabular} \label{table:eut}
\end{table*}

\begin{thm} \label{thm:eut_table}
The optimal sensing decision $B_s^*$ for problem (10) and the optimal leasing decision $B_l^\ast$ for problem (5) under EUT are summarized in Table \ref{table:eut}.
\end{thm}

\begin{proof}
In the proof, we divide the feasible range of $B_s$ into two intervals, $\left[0, \frac{D}{\alpha_1}\right]$, and $\left[\frac{D}{\alpha_1},\infty\right)$.

$1)$ Case I: $B_s\leq \frac{D}{\alpha_1}$. We first compute the optimal decision $B_s^*$ in this interval. In this case, $B_s\alpha$ is not always larger than $D$, thus by (9), we can write the expectation of revenue $R\left(B_s,B_l^*,\alpha\right)$ with respect to $\alpha$. In each sub-interval $0\leq B_s\leq \frac{D}{\alpha_I}$\footnote{In the case $0\leq B_s\leq \frac{D}{\alpha_I}$, the utility function $U\left(B_s\right)$ is equivalent to $U\left(B_s\right)$ in (40) with $j=I$.} and $\frac{D}{\alpha_{j+1}}\leq B_s\leq \alpha_j$, $j=1,...,I-1$, we can write 
        \begin{align}
U\left(B_s\right)=&\mathbb{E}_{\alpha}\left[R\left(B_s,B_l^*,\alpha\right)\right]\notag\\
=&\!\left(\pi D\!-\!B_sc_s\right)\!\!\!\!\sum\limits_{i = j+1}^I\!\!\!p_i\!+\!\!\sum\limits_{i=1}^j\!\left[\left(\pi\!-\!c_l\right)\!D\!\!-\!\!B_s(c_s\!-\!c_l\alpha_i)\right]\!p_i\notag\\
        =&\pi D-\sum\limits_{i=1}^j p_iDc_l+\left[\left(\sum\limits_{i=1}^j\alpha_ip_i\right)c_l-c_s\right]B_s.
				\end{align}

From (40), we know that $U\left(B_s\right)$ is continuous and piecewise linear in $B_s$. Since $\sum\limits_{i=1}^j\alpha_ip_ic_l$ in (41) is increasing in $j$, we know that there is at most one local maximum point. Hence, the optimal $B_s^*$ is either at the local maximum point or the boundaries, depending on the value of $\frac{c_l}{c_s}$,
	
        \begin{align}
			B_s^* &\!=\! \left\{
				\begin{aligned}
        \!0, ~&\text{if } \frac{c_l}{c_s}\leq\left(\sum\limits_{i=1}^{I}\alpha_ip_i\right)^{-1}.\\        \!\frac{D}{\alpha_j},~&\text{if} \left(\!\sum\limits_{i=1}^{j\!+\!1}\alpha_ip_i\!\!\right)^{\!-1}\!\!\!\!\!\!\!\!<\!\frac{c_l}{c_s}\!<\!\!\left(\!\sum\limits_{i=1}^{j}\alpha_ip_i\!\!\right)^{\!-1}\!\!\!\!\!\!\!\!,~ j \!=\! 1,2,...,I\!\!-\!\!1,\\
        				        \!\frac{D}{\alpha_1},~&\text{if } \frac{c_l}{c_s}\geq\left(\alpha_1p_1\right)^{-1}.\\
        \end{aligned}
				\right.
        \end{align}

From the analysis of Case III in Appendix A, we know that the optimal $B_s^*\in\left[0,\frac{Dc_l}{c_s}\right]$. By comparing the value of $\frac{Dc_l}{c_s}$ with the optimal $B_s^*$ in (40), we obtain the results in the first two rows of Table \ref{table:eut}.

$3)$ Case II: $B_s\in\left[\frac{D}{\alpha_1},\infty\right)$.\\
In this case, $B_s\alpha\geq D$. From (9), we know 
\begin{equation}
    U\left(B_s\right)=\mathbb{E}\left[R\left(B_s,B_l^*,\alpha\right)\right]=\pi D-B_sc_s.
\end{equation}
Since $R\left(B_s\right)$ is decreasing in $B_s$ in this case, the optimal sensing decision $B_s^*=\frac{D}{\alpha_1}$, and the corresponding utility $U\left(B_s^*\right)=\pi D-\frac{D}{\alpha_1}c_s$

Since $U\left(B_s\right)$ is continuous, we combine the optimal utilities from Case I and Cases II, and obtain the optimal $B_s^*$ with different values of $\frac{c_l}{c_s}$ as in Table \ref{table:eut}.
\end{proof}

The results in Table \ref{table:eut} also depend on the cost ratio $\frac{c_l}{c_s}$. When $\frac{c_l}{c_s} \leq \left(\sum_{i=1}^I p_i\alpha_i\right)^{-1}$, the leasing is cheap enough so that the operator will choose to lease only ($B_s^*=0$). For the case $\frac{c_{l}}{c_s}\geq \left(p_1\alpha_1\right)^{-1}$ (hence leasing is significantly more expensive), we have $B_s^*=\frac{D}{\alpha_1}$. The threshold value $\left(\sum_{i=1}^{j} p_i\alpha_i\right)^{-1}$ is based on the distribution of $\alpha$, as the expected ``effective cost'' of getting one unit of idle spectrum through sensing is $c_s\left(\sum_{i=1}^j p_i\alpha_i\right)^{-1}$.

\subsection{Proof of Theorem 2}
From (12), when $B_s>D$, $U\left(B_s\right)$ is decreasing in $B_s$, so $B_s^* \in \left[0, D\right]$. Hence, we obtain
\begin{align}
U\left(B_s\right)&=-\lambda w\left(p_1\right)\left(B_sc_s\right)^\beta+\left(c_l-c_s\right)^\beta B_s^\beta w\left(p_2\right)\notag\\
&=\left[w\left(p_2\right)\left(c_l-c_s\right)^\beta-\lambda w\left(p_1\right){c_s}^\beta\right]B_s^\beta.
\end{align}
From (43), we find that $U\left(B_s\right)$ is a monotonic function of $B_s$. Hence, we can find the optimal sensing decision
        \begin{align}
				B_s^* &= \left\{
				\begin{aligned}
				0,\quad&\text{if } w\left(p_2\right)\left(c_l-c_s\right)^\beta<\lambda w\left(p_1\right){c_s}^\beta,\\
        D, \quad&\text{if } w\left(p_2\right)\left(c_l-c_s\right)^\beta\geq\lambda w\left(p_1\right){c_s}^\beta.\\
        \end{aligned}
				\right.
        \end{align}

\subsection{Proof of Theorem 3}
From (13) and (14), when $B_s>D$, both $U_{RPH}\left(B_s\right)$ and $U_{RPL}\left(B_s\right)$ are decreasing in $B_s$, so $B_s^* \in \left[0, D\right]$. We prove Theorem 3 by capturing the unimodal structure of (13) and (14). For a unimodal problem, the optimal point is either at the unique local maximum point or the boundaries.
We first compute the first order derivatives of $U_{RPH}$ and $U_{RPL}$ with respect to $B_s$:
\begin{align}
\frac{\partial U_{R\!P\!H}\!\left(B_s\right)}{\partial B_s}=&-\lambda c_s w\left(p_1\right) \beta\left[B_sc_s+D\left(c_l-c_s\right)\right]^{\beta-1}\notag\\
&+\lambda w\left(p_2\right)\beta\left(c_l-c_s\right)^{\beta}\left[\left(D-B_s\right)\right]^{\beta-1},
\end{align}
and
\begin{align}
&\frac{\partial U_{R\!P\!L}\!\left(B_s\right)}{\partial B_s}=\!-c_s\beta\left(-B_sc_s+Dc_s\right)^{\beta-1}w\left(p_1\right)\notag\\
&~~+\!\beta\left(c_l\!-\!c_s\right)\left[B_s\left(c_l\!-\!c_s\right)\!+\!Dc_s\right]^{\beta-1}w\left(p_2\right).
\end{align}
Since the second order derivatives $\frac{\partial^2 U_{RPH}\left(B_s\right)}{\partial B_s^2}>0$ and $\frac{\partial^2 U_{RPL}\left(B_s\right)}{\partial B_s^2}<0$, the function $\frac{\partial U_{RPH}}{\partial B_s}$ is a strictly increasing function of $B_s$, and the function $\frac{\partial U_{RPL}\left(B_s\right)}{\partial B_s}$ is a strictly decreasing function of $B_s$, which means $\frac{\partial U_{RPH}\left(B_s\right)}{\partial B_s}$ and $\frac{\partial U_{RPL}\left(B_s\right)}{\partial B_s}$ will only be zero at most once, thus at most one local maximum point for both $U_{RPH}\left(B_s\right)$ and $U_{RPL}\left(B_s\right)$. 

We then consider the two boundary points (a) $B_s=0+\epsilon$ and (b) $B_s=D-\epsilon$, with $\epsilon$ being a small positive number approaching zero (i.e., $\epsilon \rightarrow 0$), to see if the optimal point is at the local maximum point or at the boundaries.

(a) When $B_s=0+\epsilon$, we have:
\begin{align}
\lim_{\epsilon\rightarrow0} U'_{R\!P\!H}\!\!\left(\epsilon\right)\!=\!\beta\lambda\left[D\!\left(c_l\!-\!c_s\right)\right]^{\beta-1}\!\left[w\!\left(p_2\right)\!\left(c_l\!-\!c_s\right)\!-\!w\!\left(p_1\right)\!c_s\right],
\end{align}
and
\begin{align}
\lim_{\epsilon\rightarrow0} U'_{R\!P\!L}\!\left(\epsilon\right)=\beta Dc_s^{\beta-1}\left[-w\left(p_1\right)c_s+\left(c_l-c_s\right)w\left(p_2\right)\right].
\end{align}
(b) When $B_s=D-\epsilon$, we have:
\begin{align}
\lim_{\epsilon\rightarrow0} U'_{RPH}\left(D-\epsilon\right)=\infty,
\end{align}
and
\begin{align}
\lim_{\epsilon\rightarrow0} U'_{RPL}\left(D-\epsilon\right)=-\infty.
\end{align}
We can obtain 
\begin{align}
\lim_{\epsilon\rightarrow0} U'_{RPH}\left(\epsilon\right)<0 \Leftrightarrow & -w\left(p_1\right)c_s+\left(c_l-c_s\right)w\left(p_2\right)<0 \notag\\
\Leftrightarrow&\lim_{\epsilon\rightarrow0} U_{RPL}\left(\epsilon\right)<0.
\end{align}
Since $U'_{RPH}\left(\epsilon\right)$ and $U'_{RPL}\left(\epsilon\right)$ have the same sign from (47) and (48), either $U'_{RPH}\left(B_s\right)$ is first negative then positive and $U'_{RPL}\left(B_s\right)$ is all negative within the range $\left(0,D\right)$, or $U'_{RPH}\left(B_s\right)$ is all positive and $U'_{RPL}\left(B_s\right)$ is first positive then negative within the range $\left(0,D\right)$.

Since $U\left(B_s\right)$ is continuous in $B_s\in\left[0,D\right]$, the optimal solution $B_s^*$ under the two reference points depends on the value of $c_l$ and $c_s$ as in Table \ref{table:rpp}.\\

\begin{IEEEbiography}
[{\includegraphics[width=1in,height=1.25in,clip,keepaspectratio]{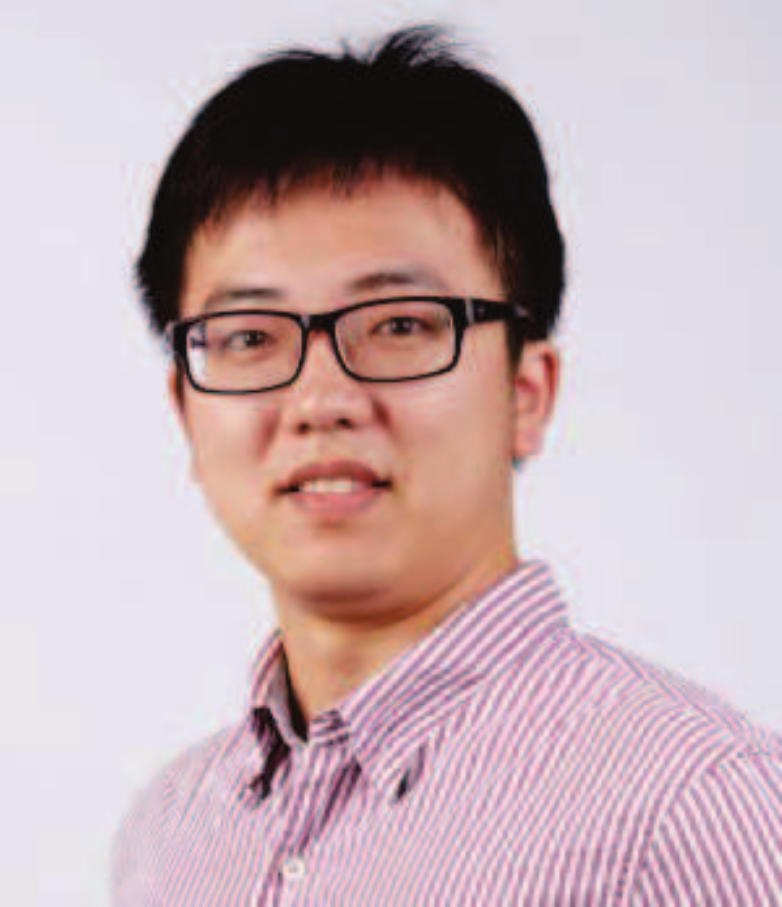}}]
{Junlin Yu} (S'14) is working towards his Ph.D. degree in the Department of Information Engineering at the Chinese University of Hong Kong. His research interests include behavioral economical studies in wireless communication networks, and optimization in mobile data trading. He is a student member of IEEE.
\end{IEEEbiography}

\begin{IEEEbiography}
[{\includegraphics[width=1in,height=1.25in,clip,keepaspectratio]{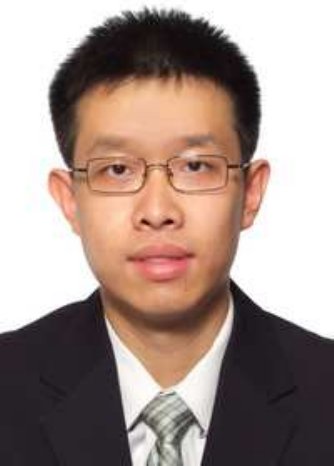}}]
{Man Hon Cheung} received the B.Eng. and M.Phil. degrees in Information Engineering from the Chinese University of Hong Kong (CUHK) in 2005 and 2007, respectively, and the Ph.D. degree in Electrical and Computer Engineering from the University of British Columbia (UBC) in 2012.
 Currently, he is a postdoctoral fellow in the Department of Information Engineering in CUHK.
 He received the IEEE Student Travel Grant for attending {\it IEEE ICC 2009}. He was awarded the Graduate Student International Research Mobility Award by UBC, and the Global Scholarship Programme for Research Excellence by CUHK.
 He serves as a Technical Program Committee member in {\it IEEE ICC}, {\it Globecom}, and {\it WCNC}.
 His research interests include the design and analysis of wireless network protocols using optimization theory, game theory, and dynamic programming, with current focus on mobile data offloading, mobile crowd sensing, and network economics.
\end{IEEEbiography}

\begin{IEEEbiography}[{\includegraphics[width=1in,height=1.25in,clip,keepaspectratio]{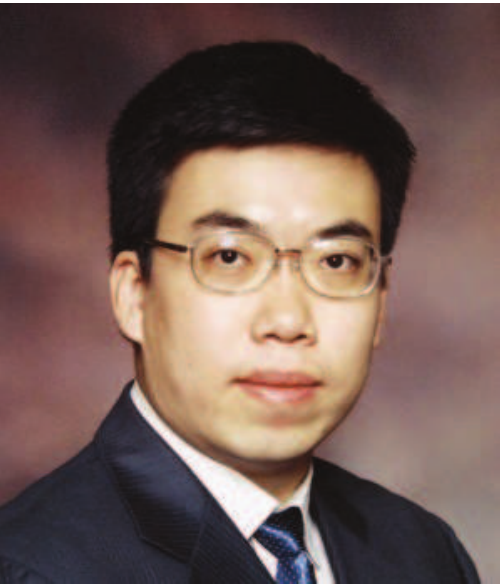}}]
{Jianwei Huang} (F'16) is an Associate Professor and Director of the Network Communications and Economics Lab (ncel.ie.cuhk.edu.hk), in the Department of Information Engineering at the Chinese University of Hong Kong. He received the Ph.D. degree from Northwestern University in 2005. He is the co-recipient of 8 Best Paper Awards, including IEEE Marconi Prize Paper Award in Wireless Communications in 2011. He has co-authored five books: \emph{Wireless Network Pricing}, \emph{Monotonic Optimization in Communication and Networking Systems}, \emph{Cognitive Mobile Virtual Network Operator Games}, \emph{Social Cognitive Radio Networks}, and \emph{Economics of Database-Assisted Spectrum Sharing}. He has served as an Associate Editor of IEEE Transactions on Cognitive Communications and Networking, IEEE Transactions on Wireless Communications, and IEEE Journal on Selected Areas in Communications - Cognitive Radio Series. He is the Vice Chair of IEEE ComSoc Cognitive Network Technical Committee and the Past Chair of IEEE ComSoc Multimedia Communications Technical Committee. He is a Fellow of the IEEE and a Distinguished Lecturer of IEEE Communications Society.
\end{IEEEbiography}

\end{document}